\documentclass[final,3p,times]{elsarticle}
\label{package}
\usepackage[numbers]{natbib}
\usepackage{stfloats}  
\usepackage{float}  
\usepackage{multirow}
\usepackage{graphicx}
\usepackage{epstopdf}
\usepackage{subfigure}
\usepackage{array}
\usepackage{soul, color, xcolor}
\usepackage{changes}
\usepackage{algorithm} 	
\usepackage[noend]{algpseudocode}
\usepackage{ntheorem}	
\usepackage{booktabs}  
\usepackage{amsmath}
\usepackage{caption}

\usepackage{placeins}
\renewcommand{\algorithmicrequire}{\textbf{Input:}}
\renewcommand{\algorithmicensure}{\textbf{Output:}}

\usepackage{amssymb}
\usepackage{ragged2e}
\usepackage{placeins}
\usepackage[colorlinks,linkcolor=blue,citecolor=blue]{hyperref}
\usepackage{float}  
\usepackage{lipsum}
\usepackage{microtype}

\theoremseparator{.}
\newtheorem{example}{Example}[section]  

\newtheorem*{proof}{Proof}
\newtheorem{lemma}{Lemma}[section]

\newtheorem{definition}{Definition}[section]

\makeatletter
\def\ps@pprintTitle{%
	\let\@oddhead\@empty
	\let\@evenhead\@empty
	\def\@oddfoot{}%
	\let\@evenfoot\@oddfoot
}
\makeatother

\begin{document}

\begin{frontmatter}

\title{Topology-Aware Subset Repair via Entropy-Guided Density and Graph Decomposition}

\author[hit]{Guoqi~Zhao}
\ead{24s030123@stu.hit.edu.cn}

\author[hit]{Xixian~Han}
\ead{hanxx@hit.edu.cn}

\author[hit]{Xiaolong~Wan\corref{cor1}}
\cortext[cor1]{Corresponding author}
\ead{wxl@hit.edu.cn}

\affiliation[hit]{organization={School of Computer Science and Technology},
            addressline={Harbin Institute of Technology}, 
            city={Harbin},
            postcode={150000}, 
            state={Heilongjiang},
            country={China}}

\begin{abstract}
Subset repair is an important data cleaning technique that enforces integrity constraints by deleting a minimal number of conflicting tuples, yet multiple minimal repairs often exist. Density-based methods address this ambiguity by favoring repairs that preserve dense, high-quality data regions; however, their effectiveness is limited by density bias from dirty clusters, high computational cost, and uniform attribute weighting.
We propose a topology-aware approximate subset repair framework based on a joint density–conflict penalty model. The framework integrates three key components. First, a two-layer conflict detection strategy combines attribute inverted indexes with CFD rule grouping to efficiently identify violations. Second, we introduce EntroCFDensity, a density metric that incorporates information entropy and CFD weights to dynamically adjust attribute importance and reduce homogeneity bias. Third, a conflict degree measure is defined to complement local density, enabling a topology-adaptive penalty mechanism with dynamic weight allocation guided by the coefficient of variation. The conflict graph is further decomposed into independent subgraphs, transforming global repair into tractable local subproblems.
Based on this framework, we develop two algorithms: PPIS, a scalable heuristic, and MICO, a mixed-integer programming method with theoretical guarantees. Experimental results show that our approach improves repair accuracy and robustness while effectively preserving high-quality data.
\end{abstract}


\end{frontmatter}


\section{Introduction}\label{sec:introduction}

In real-world databases, violations of integrity constraints (such as functional dependencies) are common, leading to data redundancy, logical errors, and ultimately undermining the reliability of data-driven decisions. Integrity constraints formally characterize such inconsistencies and enable systematic detection of tuples that conflict with the intended data semantics~\cite{DBLP:journals/pvldb/MiaoCLGL20}. Within this context, subset repair (S-repair) has become a classic framework for database repair, resolving inconsistencies by removing tuples so that the remaining subset satisfies all given constraints~\cite{DBLP:journals/pvldb/LiuSGGKR24}. By addressing conflicts with minimal deletion cost, S-repair preserves as much useful information as possible while maintaining semantic consistency. This balance makes S-repair essential in data cleaning, as it effectively restores consistency while minimizing information loss, thereby enhancing the usability and trustworthiness of the cleaned data~\cite{DBLP:journals/tkde/SunSY24,DBLP:journals/tods/LivshitsKR20}.

Specifically, if a database instance $I$ violates a given set of integrity constraints, it is considered inconsistent. A common type of integrity constraint is the functional dependency (FD)~\cite{DBLP:conf/ifip/Armstrong74}. Given a relation schema with attribute set $U$, an FD is expressed as $\varphi: X \rightarrow Y$ where $X, Y \subseteq U$. It requires that for any two tuples $t_1, t_2$ in $I$, if $t_1[X] = t_2[X]$, then $t_1[Y] = t_2[Y]$ must hold. As an extension, the Conditional Functional Dependency (CFD)~\cite{DBLP:conf/icde/BohannonFGJK07} enhances the semantics of FDs by introducing a pattern tableau that binds specific attribute values. This allows CFDs to express both general functional dependencies and value-specific semantic constraints, thereby capturing more granular data rules that are common in real-world applications. For a relation instance $I$ that violates a set of constraints $\Sigma$, the goal of a subset repair (S-repair) is to remove a set of tuples $I_n \subseteq I$ such that the remaining instance $I_S = I - I_n$ satisfies $\Sigma$. The repair aims to achieve this with minimal deletion cost—often interpreted as removing the fewest tuples or minimizing a given cost function over the tuples. 

\begin{table}[thb]
\centering
\caption{employee dataset}
\small 
\begin{tabular}{|c|c|c|c|c|c|}
\hline
Id & Work experience & Salary & Position & Allowance & Data labeling \\
\hline
$t_{0}$ & 1 & 6500 & manage     & 1000 & clean \\
$t_{1}$ & 1 & 6500 & manage     & 1000 & clean \\
$t_{2}$ & 1 & 6600 & manage     & 1000 & dirty \\
$t_{3}$ & 1 & 6600 & operate    &  600 & dirty \\
$t_{4}$ & 3 & 7500 & technology &  800 & clean \\
$t_{5}$ & 3 & 7500 & technology &  800 & clean \\
$t_{6}$ & 3 & 8000 & operate    &  800 & dirty \\
$t_{7}$ & 3 & 8000 & operate    &  800 & dirty \\
$t_{8}$ & 3 & 8000 & operate    &  700 & dirty \\
$t_{9}$ & 3 & 8000 & operate    & 1000 & dirty \\
\hline
\end{tabular}
\label{employee dataset}
\end{table}

\begin{example}
\normalfont
Consider the example dataset in Table~\ref{employee dataset}, under the constraints of the functional dependency set $\Sigma = \{ FD_1: \text{Work experience} \rightarrow \text{Salary}, FD_2: \text{Position} \rightarrow \text{Allowance} \}$, every tuple violates at least one dependency, meaning that each tuple is conflicting. To restore consistency, one possible solution is to remove the tuple set $I_n = \{ t_{2}, t_{3}, t_{6}, t_{7}, t_{8}, t_{9} \}$, after which the remaining subset $\{ t_{0}, t_{1}, t_{4}, t_{5} \}$ satisfies all constraints in $\Sigma$. Thus, $I_n$ is a valid minimal removal set. However, this solution is not unique. Other removal sets (such as $\{ t_{2}, t_{3}, t_{4}, t_{5}, t_{8}, t_{9} \}$ or $\{ t_{0}, t_{1}, t_{6}, t_{7}, t_{8}, t_{9} \}$) also yield a consistent subset. These observations highlight two key challenges in subset repair: (1) the minimal removal set is generally not unique, often leading to a large set of candidate repairs; and (2) it is difficult to ensure that the selected repair best preserves high-quality data, since the "best" choice depends on the cost metric tailored to downstream tasks~\cite{DBLP:journals/pvldb/MiaoCLGL20}.
\end{example}

In terms of the repair objectives, current subset repair approaches can be primarily categorized into two types: (1) Minimum-cardinality subset repair~\cite{DBLP:journals/iandc/ChomickiM05,DBLP:journals/tods/LivshitsKR20} , which aims to remove the fewest number of tuples such that all constraints are satisfied. Although this approach helps avoid excessive data loss, it ignores differences in data quality and often yields multiple valid solutions due to the combinatorial nature of conflict resolution. For these reasons, this type is not considered further in our study. (2) Minimum-weight subset repair, where each tuple is assigned a weight representing its importance or reliability—such as cost~\cite{DBLP:journals/pvldb/LiuSGGKR24,DBLP:journals/pvldb/MiaoCLGL20,DBLP:journals/vldb/MiaoZLWC23,DBLP:journals/tods/LivshitsKR20} , probability~\cite{DBLP:conf/apweb/ZhangHZLX23} , or density~\cite{DBLP:journals/tkde/SunSY24} . The goal is to minimize the total weight of deleted tuples while ensuring constraint satisfaction. This allows higher-value tuples to be preserved and can better reflect real-world priorities. However, the effectiveness of this approach heavily depends on how weights are defined, and the optimization problem generally becomes computationally harder. Importantly, ~\cite{DBLP:journals/tods/LivshitsKR20} established a dichotomy theorem showing that the S-repair problem is solvable in polynomial time only under specific conditions—such as when attributes share a common left-hand side, a consensus functional dependency exists, or the LHS marriage condition holds. In general, the problem is NP-hard. As a result, practical methods must rely on heuristics or approximation algorithms. A notable instance of the minimum-weight S-repair framework is density-based S-repair, but it also inherits the limitations of weight-based methods. 

Intuitively, we refer to tuples whose values are correct with respect to the (unknown) ground truth as \emph{clean tuples}. ~\cite{DBLP:journals/tkde/SunSY24}  use $k$-nearest-neighbor density as a quantitative measure of tuple reliability and select a subset repair that maximizes the overall density of the remaining database. By transforming the problem into one of density maximization, it extends the traditional minimum-weight deletion concept and aims to more accurately identify true errors. However, despite this innovation, existing density-based, and more broadly, weight-based S-repair methods still face several critical limitations: (1) Homogenization of attribute weights. The density function implicitly treats all attributes as equally important, ignoring differences in their frequency and informativeness in CFD rules. As a result, attributes that carry more semantic constraint information do not receive greater influence in density estimation. This motivates the need for an importance-aware density measure that can automatically adjust attribute weights. (2) Density bias under clustered dirty data. When erroneous tuples form local clusters, their neighborhood density may exceed that of clean tuples, causing the repair algorithm to misinterpret clustered noise as "high-density, reliable" data. This reveals the weakness of relying solely on local density, which fails to capture conflicts in the global topology. This motivates introducing a global conflict-topology metric to mitigate such density distortions. (3) High computational cost. The conflict detection phase requires a time complexity of $O(n^2 \times r)$ (where $n$ is the tuple cardinality and $r$ is the average number of rule checks), severely limiting its scalability on large-scale datasets. This bottleneck calls for more efficient conflict detection and conflict-graph decomposition techniques to enable large-scale repair.

To address the limitations of existing S-repair methods, we propose a topology-aware approximate S-repair framework centered on a density–conflict degree penalty score. The framework integrates three key ideas: a reliability-aware density measure, a global conflict-aware correction term, and an adaptive weighting mechanism. First, to mitigate attribute-weight homogenization in traditional density estimation, we introduce EntroCFDensity, a new tuple reliability metric that integrates attribute entropy with attribute frequencies in CFD constraints. Second, since density alone is vulnerable to dirty-data clustering, we incorporate the conflict degree, defined as the number of tuples violating CFDs with a given tuple, which captures global conflict interactions beyond local neighborhoods. Third, by using the coefficient of variation to characterize the distribution of density and conflict degree within each component, the framework adaptively produces topology-aware weights and constructs the Density–Conflict Degree Penalty Score. This penalty score provides systematic bias correction: tuples with low density but also low conflict degree are protected from erroneous deletion, whereas tuples with both high density and high conflict degree receive amplified penalties and are prioritized for removal. To support large-scale repair, we further design an efficient conflict detection strategy based on inverted indexes and rule grouping. By reducing the repair task to finding a minimum vertex cover with minimal penalty score in the conflict graph, the problem is NP-hard. Based on this framework, we design two algorithms: PPIS, an efficient and fast heuristic method, and MICO, a mixed-integer programming approach with theoretical guarantees.

The contributions of this paper are summarized as follows.
\begin{itemize}
	\item We propose a topology-aware density–conflict penalty model that enables cost-minimal subset repair under CFDs.
	\item We develop an efficient conflict detection and conflict-graph construction method, and formulate the repair task as a penalty-aware minimum vertex cover problem.
	\item We establish a penalty-driven approximate subset repair framework and design two algorithms, PPIS and MICO, to balance efficiency and optimality.
	\item We conduct extensive experiments demonstrating the advantages of the proposed methods in both repair quality and computational efficiency.
\end{itemize}

The structure of this paper is organized as follows.
Section 2 reviews related work in data cleaning. 
Section 3 introduces the data model and constraint definitions. 
Section 4 describes the conflict detection strategy, conflict-graph construction, and the proposed density–conflict penalty mechanism. 
Section 5 presents the overall repair framework and the two algorithms, PPIS and MICO. 
Section 6 evaluates the methods experimentally. 
Section 7 concludes the paper and discusses future work. 

\section{Related Work}\label{sec:relatedWork}
Dirty data has long been recognized as a key challenge in data management and analysis. Real-world data is often affected by quality issues, and dirty data may lead to incorrect clustering results or misidentification of patterns and relationships within the data ~\cite{DBLP:journals/tkde/LiangSSL25}. Studies have shown that data-driven decisions based on noisy, missing, or erroneous data can sometimes be even less reliable than those made without data ~\cite{DBLP:journals/pvldb/BaoBFLLLLLLLOSTWWWXZZZZ24}. Consequently, data cleaning has become particularly critical. Researchers in both database (DB) and machine learning (ML) communities have paid significant attention to data cleaning issues and actively explored various related research topics~\cite{DBLP:conf/icde/LiRBZCZ21}.

The data cleaning process generally consists of error detection and error repair~\cite{DBLP:conf/icde/LiRBZCZ21}. To achieve these objectives, researchers have proposed various types of data quality rules and conducted systematic studies on improving both the effectiveness and efficiency of data cleaning algorithms. Existing approaches can be broadly categorized into three classes: rule-based data cleaning, ML (Machine Learning)–based data cleaning, and statistical probability-based data cleaning.

\subsection{Rule-based data cleaning}
The basic idea of rule-based data cleaning is to ensure database consistency through a predefined set of constraints, thereby achieving data consistency and improved data quality ~\cite{DBLP:journals/pvldb/GeertsMPS13}. Data quality rules (e.g., integrity constraints) provide a declarative means of specifying valid or correct data instances, and any data that violates these rules is regarded as an error, or a violation~\cite{DBLP:journals/ftdb/IlyasC15}. Rule-based data cleaning involves correcting inconsistent data by leveraging these rules or constraints~\cite{DBLP:journals/pvldb/GeertsMPS13}.

Existing rule-based data repair techniques focus on computing repairs, with the goal of minimizing modifications to the data instance while satisfying a set of rules.  Representative rules include functional dependencies (FDs) ~\cite{DBLP:conf/sigmod/DallachiesaEEEIOT13,DBLP:journals/tkde/Hao0LHTF17,DBLP:journals/pvldb/PenaAN21}, conditional functional dependencies (CFDs) ~\cite{DBLP:journals/tkde/SunSY24}, relaxed functional dependencies (RFDs) ~\cite{DBLP:conf/edbt/BreveCDP22}, denial constraints ~\cite{DBLP:journals/pvldb/PenaAN21}, and repair rules  ~\cite{DBLP:conf/sigmod/WangT14}.

\textbf{Minimum-cardinality subset repair. }~\cite{DBLP:journals/iandc/ChomickiM05} conducted a systematic study on subset repair with minimal cardinality, providing a detailed analysis of the computational complexity of repair checking and consistent query answering. The work established clear boundaries between PTIME and co-NP-hard (and beyond) for different classes of constraints. It further introduces the Conflict Hypergraph to model tuple-level conflicts, proving that database repairs correspond bijectively to maximal independent sets of this hypergraph, thereby reducing the repair problem to a classically hard problem in graph theory. ~\cite{DBLP:journals/tods/LivshitsKR20} propose OptSRepair, which transforms the search for an optimal subset repair into an equivalent problem of finding a minimum-weight vertex cover in the conflict graph. The method  provides exact polynomial-time solutions for specific types of FD sets. Additionally, the study proved that for other FD sets, even finding a cardinality-minimal subset repair—a special case where all weights are 1—is APX-complete. The results fully cover the important special case of minimum-cardinality repair.

However, minimum-cardinality subset repair suffers from two main limitations. First, there may exist multiple repair solutions without a clear criterion for determining the optimal one. Second, it does not incorporate tuple weights or priority information, instead assuming that all tuples have equal weight. In practical scenarios, it may be more reasonable to remove one low-weight tuple rather than remove a single high-weight tuple. Therefore, solely pursuing minimum-cardinality may fail to capture actual needs. It is necessary to consider minimum-weight repair, which integrates the relative value differences among tuples into the repair process.

\textbf{Minimum-weight subset repair. }~\cite{DBLP:journals/pvldb/LiuSGGKR24} investigated the cost issue of introducing Representation Constraints (RCs) in S-repair. It formally defined the concept of a Representative Subset Repair (RS-repair), which requires satisfying Functional Dependencies (FDs) while ensuring the proportion of sensitive attributes adheres to predefined constraints. The cost is defined as the cost of representation, quantified by the deletion overhead—i.e., the proportion of additional tuples that must be deleted in an optimal RS-repair compared to an optimal S-repair. RS-repair helps maintain group representativeness, addressing the problem where traditional repair methods may exacerbate data bias. ~\cite{DBLP:journals/pvldb/MiaoCLGL20} studied the computational problem of finding constraint-based Optimal Subset Repairs (OPTSR), where the repair cost is defined as the weighted sum of deleted tuples. Here, the weights can represent tuple confidence, query contribution, or benefit, thereby quantifying the repair cost. The work extensively explored the application of OPTSR in multiple scenarios: assessing the degree of inconsistency in data repair, supporting the computation of upper and lower bounds for aggregate queries in Consistent Query Answering, and guiding decision-making in optimization problems. The authors proposed tighter inapproximability lower bounds, proving that OPTSR is often hard to approximate within a factor of 17/16, and developed three approximation algorithms, among which the TE-LP algorithm achieves a ratio better than 2. Subsequent research ~\cite{DBLP:journals/vldb/MiaoZLWC23} has established tighter inapproximability bounds for OPTSR, such as 143/136 for most cases.

From the perspective of error detection, ~\cite{DBLP:conf/icde/SunSW020} points out that tuples containing misplaced attribute values typically exhibit significant numerical differences, that is, greater distances from other tuples. This anomalous distance (or similarity) serves as a crucial signal for identifying potential misplacement errors. Building on the findings, ~\cite{DBLP:journals/tkde/SunSY24} introduces density as an optimization metric for subset repair. Density is defined as the sum of the reciprocal distances between a tuple and its k nearest neighbors, and is used to characterize the cleanliness of tuples. Based on the observation that cleaner tuples generally exhibit higher density, the subset repair problem is reformulated as a density maximization problem. The objective is to find a removal set such that the remaining instance achieves maximum density while satisfying all given constraints. By modeling the set of conflicting tuples as a conflict graph, the density maximization problem is reduced to a weighted vertex cover problem, whose NP-hardness is established. Efficient heuristic algorithms and the RELAXATION algorithm with performance guarantees are then developed, providing approximation solutions and error bounds for the NP-hard S-repair optimization problem.

However, this density-based S-repair also has limitations. First, it assumes homogeneous attribute weights, treating all attributes as equally important and neglecting the heterogeneity in their informational value. Second, its core assumption that clean tuples exhibit higher density within their neighborhoods, may not always hold in real-world data, as aggregated dirty data can lead to biased density estimates. Additionally, the method presents a significant computational challenge, as the conflict detection stage entails a substantial computational cost.

\textbf{Update repair (U-repair).} The rule-based U-repair model removes violations by updating the values of tuples or cells. Holistic ~\cite{DBLP:conf/icde/ChuIP13} encodes constraint violations as a Conflict Hypergraph, employs a Minimum Vertex Cover (MVC) to identify potentially erroneous cells, and constructs a Repair Context (RC) to collect candidate repair expressions; finally, it employs an iterative algorithm to dynamically update cell values for data repair. ~\cite{DBLP:journals/tkde/KoehlerL22} proposed a qualitative cleaning model based on possibility theory, reformulating the problem as a possibilistic vertex cover task and providing a fixed-parameter tractable algorithm. This approach achieves non-intrusive update repairs and effectively addresses constraint repair challenges under uncertain data. Horizon ~\cite{DBLP:journals/pvldb/RezigOAEMS21} constructs FD pattern graphs to encode value combination patterns within the data and selects repairs based on pattern frequency rather than the minimal change principle. However, as the underlying data may not necessarily contain erroneous values, this approach can lead to over-repair.

Despite the respective advantages of the aforementioned methods, U-repair still suffers from three inherent limitations: (1) It requires generating new cell values to resolve conflicts, yet these values often lack sufficient semantic justification, potentially introducing unverifiable or even new errors. (2) The candidate update space for U-repair typically grows exponentially, posing a significant computational burden for the repair search process, especially under complex constraints or high data noise scenarios. (3) Update operations may trigger chain reactions, where modifying one cell affects multiple constraints, thereby exacerbating the risk of error propagation.

To address these limitations, we adopt S-repair as the repair paradigm. On the one hand, S-repair addresses conflicts solely through tuple deletion, avoiding the semantic uncertainty introduced by fabricating new values and yielding repair outcomes with higher determinacy and interpretability. On the other hand, by employing the density–conflict degree penalty score as the optimization objective, the repair space becomes clearly quantifiable, offering a more stable and controllable optimization structure for designing efficient, near-optimal repair strategies.

\subsection{ML-based data cleaning}
To address the diverse characteristics of dirty data in real-world datasets, comprehensive ML-based error detection methods have been developed, which typically formulate error detection as a classification task. For example, SAGED ~\cite{DBLP:conf/edbt/0001KSS24} introduces a two-stage framework that combines meta-learning and semi-supervised learning, enabling efficient error detection under few-shot conditions and demonstrating strong generalization capability across domains. RAHA ~\cite{DBLP:conf/sigmod/MahdaviAFMOS019} generates feature vectors using simple error detection methods and then trains a detection classifier. Picket ~\cite{DBLP:journals/vldb/LiuZR22} employs a self-supervised learning model, PicketNet, to detect corrupted data during both training and inference. HoloDetect ~\cite{DBLP:conf/sigmod/HeidariMIR19}, based on few-shot learning, reduces the need for manual annotations through data augmentation and captures errors using multi-granularity feature representations. However, it still requires a certain amount of initial labeled data and may face scalability issues when applied to large datasets. In this context, Neutatz et al. proposed ED2 ~\cite{DBLP:conf/cikm/NeutatzMA19}, an active learning-based method that adopts a two-stage sampling strategy: first selecting the most promising columns, and then choosing the most informative cells within those columns for labeling. This approach significantly reduces the labeling effort, and on large-scale datasets, ED2 requires an order of magnitude fewer labels than HoloDetect. Machine learning-based methods primarily focus on cell-level error detection or value repair, producing error labels or predicted values rather than constructing a feasible data instance under constraints. So, they cannot be directly applied to the constraint-based subset repair problem discussed in this paper.

\subsection{Probabilistic and statistical data cleaning}
Probabilistic and statistical approaches to data cleaning have attracted considerable attention due to their ability to effectively model data distributions and error patterns. These methods can be broadly categorized into two groups: Bayesian inference frameworks ~\cite{DBLP:conf/icde/QinHWZZM0O024,DBLP:conf/aistats/LewASM21,DBLP:journals/ml/TsamardinosBA06} and lightweight graphical models ~\cite{DBLP:journals/pvldb/TzoumasDJ11}, both of which leverage probabilistic reasoning to alleviate the strong independence assumptions required by traditional approaches. BClean ~\cite{DBLP:conf/icde/QinHWZZM0O024} adopts a two-stage Bayesian inference framework, where a compensation scoring model mitigates the "error amplification" problem that arises when constructing networks from dirty data. It further allows network repair to incorporate simple constraints, such as regular expressions, without requiring PPL coding. PClean ~\cite{DBLP:conf/aistats/LewASM21} focuses on domain customization through probabilistic programming languages. By employing the Chinese Restaurant Process (CRP) to model implicit objects and link structures in relational databases, and combining Sequential Monte Carlo (SMC) with dynamic optimization proposal compilation, it scales to cleaning tasks on millions of records. MMHC ~\cite{DBLP:journals/ml/TsamardinosBA06} constructs a skeleton network without false negatives using the MMPC algorithm, identifies candidate parents through a max-min association heuristic, and eliminates false positives with conditional independence tests, complemented by symmetric checks to reduce edge misclassification. In contrast, lightweight graphical models ~\cite{DBLP:journals/pvldb/TzoumasDJ11} quantify attribute dependencies using mutual information, and construct a maximum spanning tree to capture nonlinear topological structures. For join predicates, they validate correlations between attributes via group-aggregate SQL queries. Although probabilistic inference methods can reconstruct the data distribution or infer potentially correct values, their fundamental repair mechanism remains rooted in value imputation, falling under the paradigm of value-repair or U-repair.

\subsection{Summary and Positioning}
Our work builds upon established concepts in data management and graph theory. Information entropy ~\cite{DBLP:journals/tkde/ChenCCHP11} and the frequency of attribute~\cite{DBLP:journals/tkde/DingWSWLG22} occurrences in constraints have been widely used to quantify data uncertainty, thereby supporting data quality assessment and cleaning tasks. Similarly, conflict graphs and vertex cover are classical graph-theoretic techniques for resolving constraint violations in databases and have long played a key role in tasks such as subset repair.

However, to the best of our knowledge, no existing work has systematically combined the following components within a unified cost-based S-repair framework to address the three major challenges discussed in Section 2.1—namely, homogeneous attribute weights, biased density estimation caused by clusters of dirty data, and high computational complexity:
(1) an attribute weighting density  based on information entropy and CFD frequency to mitigate attribute homogeneity;
(2) the introduction of conflict degree as a penalty term to correct the bias of local density estimation in dirty-data clusters;
(3) modeling constraint violations as a conflict graph, reducing the problem to vertex cover, and decomposing the global conflict graph into independent connected components to enable a divide-and-conquer approach, thereby reducing computational cost.

To address these challenges, we propose an optimization framework that systematically integrates local density evaluation, global conflict structure modeling, and a graph-based divide-and-conquer strategy at the connected-component level. By means of a density–conflict degree penalty score, our method enables efficient approximate optimal S-repair. We elaborate on our approach in the sections that follow.

\section{Problem definition}\label{sec:problemDefinition}
\begin{definition}[Functional Dependency (FD)]Let $R(U)$ be a relation schema, and $X$ and $Y$ be subsets of the attribute set $U$. If for any two tuples $t_1$, $t_2$ in an instance $I$ of $R$, whenever $t_1[X] = t_2[X]$ holds, it must be the case that $t_1[Y] = t_2[Y]$ holds, then $Y$ is said to be functionally dependent on $X$, denoted as $X \rightarrow Y$. The formal definition is as follows:
\[
\forall\ t_1, t_2 \in I,\ t_1[X] = t_2[X] \Rightarrow t_1[Y] = t_2[Y]
\]
\end{definition}

\begin{example}
\normalfont
For example, in the employee dataset of Table~\ref{employee dataset}, we have the following FDs:
\[
\text{FD}_1: \text{Work experience} \rightarrow \text{Salary}, 
\text{FD}_2: \text{Position} \rightarrow \text{Allowance}
\]
\end{example}

\begin{definition}[Conditional Functional Dependency (CFD)]CFDs extend standard FDs by incorporating semantic patterns. They are expressed as $\varphi = (X \rightarrow Y, T_p)$, where $X \rightarrow Y$ is a standard FD over schema $R$, and $T_p$ is a pattern tableau defined over $R$. Each pattern tuple $ t_p \in T_p $ contains either constants or wildcards “\_” for attributes in $X \mathbin{\cup} Y$. For a relation instance $I$, the CFD $\varphi = (X \rightarrow Y, T_p)$ holds if and only if for every pair of tuples $t_1, t_2 \in I$ and for every pattern tuple $ t_p \in T_p $ : if $t_1[X] \approx t_p[X]$ and $t_2[X] \approx t_p[X]$, then it must also be the case that $t_1[Y] \approx t_p[Y]$ and $t_2[Y] \approx t_p[Y]$. Otherwise, the CFD is violated. Here, $\approx$ denotes constant matching or wildcard matching with the pattern tuple.
\end{definition}

\begin{example}
\normalfont
For example, in the employee dataset shown in Table~\ref{employee dataset}, there exists a CFD which indicates that the allowance for all employees with position ``manage'' is fixed to 1000: $\{\text{Position} \rightarrow \text{Allowance}, (\text{manage}, 1000)\}$.
\end{example}

\begin{definition}[Conflict Tuple]A tuple that violates at least one given FD, CFD, or other integrity constraint is referred to as a conflicting tuple. We record the set of conflict tuples as $I_C$.
\end{definition}

\begin{definition}[Conflict Graph]A mathematical model used to describe the relationships of constraint violations between tuples in a database. Given a database instance $I$ and a set of constraints $\Sigma$, the conflict graph is defined as an undirected graph $G=(V, E)$,  where each vertex $v_i \in V$ corresponds to a conflicting tuple $t_i \in I_C$, and an edge $(v_i, v_j) \in E$ exists if and only if $t_i$ and $t_j$ violate at least one constraint in $\Sigma$.
\end{definition}

\begin{definition}[Connected Components of Conflict Graph]
Given a conflict graph $G = (V, E)$, its connected components are defined as maximal connected subgraphs $G_C = (V_C, E_C)$ that satisfy the following conditions:
(1) $V_C \subseteq V$ and $E_C \subseteq E$;
(2) for any $u, v \in V_C$, there exists a path $P(u, v)$;
(3) there does not exist a larger subset $V_C' \supset V_C$ such that $V_C'$ also satisfies the above conditions. Connected components are categorized into two types: clique components and non-clique components.
\end{definition}

\begin{definition}[Clique Component]
Given a connected component of the conflict graph $G_C = (V_C, E_C)$, if it satisfies the property of a complete graph:$\forall u,v \in V_C,\ (u,v) \in E_C$, then $G_C$ is called a clique component. That is, a clique component is a complete subgraph of the conflict graph, in which any two vertices are directly connected, and every pair of tuples within the clique are in conflict.
\end{definition}

\begin{definition}[Non-Clique Component]
Given a connected component of the conflict graph $G_C = (V_C, E_C)$, if there exist two non-adjacent vertices such that $\exists u, v\in V_C\quad\textrm{s.t.}\ (u,v)\notin E_C$,then $G_C$ is called a non-clique component.
In other words, a non-clique component is a connected but non-complete subgraph of the conflict graph.
\end{definition}

\begin{definition}[Conflict Degree]
The conflict degree of a tuple $t_i$ is defined as the number of tuples that conflict with it:
$$
CD(t_i) = |\{t_l \in I_C \mid (t_i, t_l) \not\models \Sigma\}|
$$
In the conflict graph $G=(V,E)$, the degree of conflict $CD(t_i)$ of a tuple $t_i$ is equal to the degree of vertex $v_i$ :
$$CD(t_i) = \deg(v_i)$$

A higher conflict degree indicates that the tuple is involved in more extensive conflicts and more severely compromises the overall data consistency.
\end{definition}

\begin{definition}[Minimum-weight Subset Repair]
Given a conflict graph $G = (V, E)$, where the vertex set $V$ represents tuples and the edge set $E$ represents conflict relations. Each vertex $v \in V$ is assigned a non-negative weight $w(v)$, which can be flexibly defined according to specific application requirements, such as data reliability and other domain-specific metrics or cost. The objective of the minimum-weight subset repair is to identify a removal set $I_n \subseteq V$ such that the remaining vertex set $I_S = V \setminus I_n$ satisfies the given set of conditional functional dependencies $\Sigma$, while minimizing the total weight of the removed vertices $\sum_{v \in I_n} w(v)$. Equivalently, this maximizes the total weight of the consistent subset $I_S$ retained in the database.
\end{definition}

\section{Conflict Modeling and Adaptive Penalty Mechanism}\label{sec:algorithm}
This section aims to elucidate the core foundation of the subset repair framework proposed in this paper, focusing on the implementation of conflict detection, the graph-structured modeling of conflict relationships, and the adaptive quantification strategy for tuple quality. Firstly, given the constraints of CFDs, we propose a conflict detection mechanism based on attribute inverted index and CFDs grouping. This mechanism innovatively integrates a dual optimization strategy: the attribute inverted index significantly reduces the candidate tuple comparison space, while the CFDs grouping mechanism minimizes redundant computations caused by repeated constraint matching. These optimizations create a cascading effect, collectively reducing detection complexity by both narrowing the search space and decreasing constraint matching overhead.

Subsequently, the detected conflict tuples are modeled as a conflict graph, where tuples serve as nodes and conflict relationships as edges. This graph structure expresses the conflict relationships between tuples, providing a structured representation for subsequent repair. Based on this structure, we further introduce an adaptive penalty scoring mechanism. By analyzing attribute information entropy, the frequency of attribute occurrence in CFDs, and the conflict degree, we construct a quality scoring function tailored for the repair task. This function transforms data quality and repair cost into quantifiable penalty scores.

In summary, this section will sequentially introduce the implementation of conflict detection, the graph-structured modeling of conflict relationships, and the adaptive quantification strategy for tuple quality, thereby laying the foundation for the repair algorithm presented in Section~5.

\subsection{Conflict Detection Based on Attribute Inverted Index and CFDs Grouping}
Conflict detection, as the first step of S-repair, aims to identify all tuple pairs that violate CFDs. Traditional conflict detection methods typically require checking all possible tuple pairs ($O(n^2)$) and their matching with all CFD rules ($O(R)$), resulting in an overall complexity of $O(n^2 \times R)$. In S-repair, it is sufficient to determine whether a conflict exists between tuples, without precisely identifying the specific CFDs they violate. Therefore, once a tuple pair is found to violate any CFD, it can be marked as a conflict pair, and the checking of the remaining CFDs for that pair can be terminated, reducing the complexity to $O(n^2 \times r)$, where $r$ represents the average number of rules checked before a conflict is detected. However, the quadratic dependency on tuple pairs still constitutes a significant performance bottleneck.

To address this issue, we propose a two-layer optimized conflict detection framework that integrates attribute-level inverted indexing and CFDs grouping. Through coordinated optimizations at both the attribute and rule levels, it effectively reduces the generation and matching overhead of candidate tuple pairs, thereby enhancing detection efficiency in large-scale data.

\textbf{Attribute Inverted Index. }By constructing inverted indexes for attribute values, candidate tuple pairs that potentially share equal values on at least one attribute can be efficiently filtered. The core advantage of this approach lies in systematically avoiding redundant comparisons between tuples whose values differ on all attributes, thereby substantially reducing invalid detection operations and improving the efficiency of conflict detection.

The performance improvement brought by inverted indexes depends heavily on the data distribution characteristics. In scenarios with highly dispersed data, the scale of the candidate conflict tuple pair set $C$ approaches a constant level, thereby significantly reducing the time complexity of conflict detection from the original $O(n^2 \times r)$ to $O(n \times (A + C \times g))$, where $n$ represents the number of tuples, $A$ represents the number of attributes, $g$ represents the number of rule groups ($g << r$), and $C$ represents the average size of the deduplicated candidate set. Candidate set deduplication is achieved using hash sets, ensuring that $C \leq n$ holds strictly even in extreme cases. In extreme scenarios, such as when all tuples have identical values across all attributes, the candidate set $C$ reaches its theoretical upper limit $n$, and the complexity of conflict detection degenerates to the pre-optimization level of $O(n^2 \times r)$.

However, on real-world datasets, the data distribution typically exhibits high dispersion, and the candidate set size $C$ is much smaller than the total number of tuples $n$. Consequently, the optimization strategy based on inverted indexes usually leads to a significant improvement in conflict detection efficiency in practice.

\textbf{CFDs Grouping. }The traditional single-rule detection employs a linear independent processing architecture, treating $m$ CFDs as entirely isolated detection units, which leads to a dual efficiency deficiency: First, it suffers from redundant computation—even when rules share identical left-hand side attribute sets (the $X$ attributes), the system must still perform completely identical attribute value comparisons, memory access operations, and conditional judgment processes separately for each rule, resulting in inflated computational complexity. Second, it results in limited index utilization, as attribute matching results cannot be reused across different rules. To address this performance bottleneck, we introduce a CFDs grouping mechanism: by merging CFDs with identical left-hand side attributes into a single group, comparisons on the same left-hand side attributes for each tuple pair need to be performed only once, rather than redundantly for each individual CFD.

\begin{example}
\normalfont
Suppose we need to detect a tuple pair $(t_1, t_2)$ against six CFDs: $\{\mathrm{CFD}_1: A,B \to C; \mathrm{CFD}_2: A,B \to D; \mathrm{CFD}_3: A,B \to E; \mathrm{CFD}_4: A,B \to F; \mathrm{CFD}_5: X \to Y; \mathrm{CFD}_6: X \to Z\}$. As shown in Figure~\ref{f4-1}, without grouping, conflict detection requires six separate rule comparisons, whereas after applying CFDs grouping strategy, only two comparisons are needed.
\end{example}

\begin{figure}[thb]               
	\centering
	\includegraphics[scale=0.5]{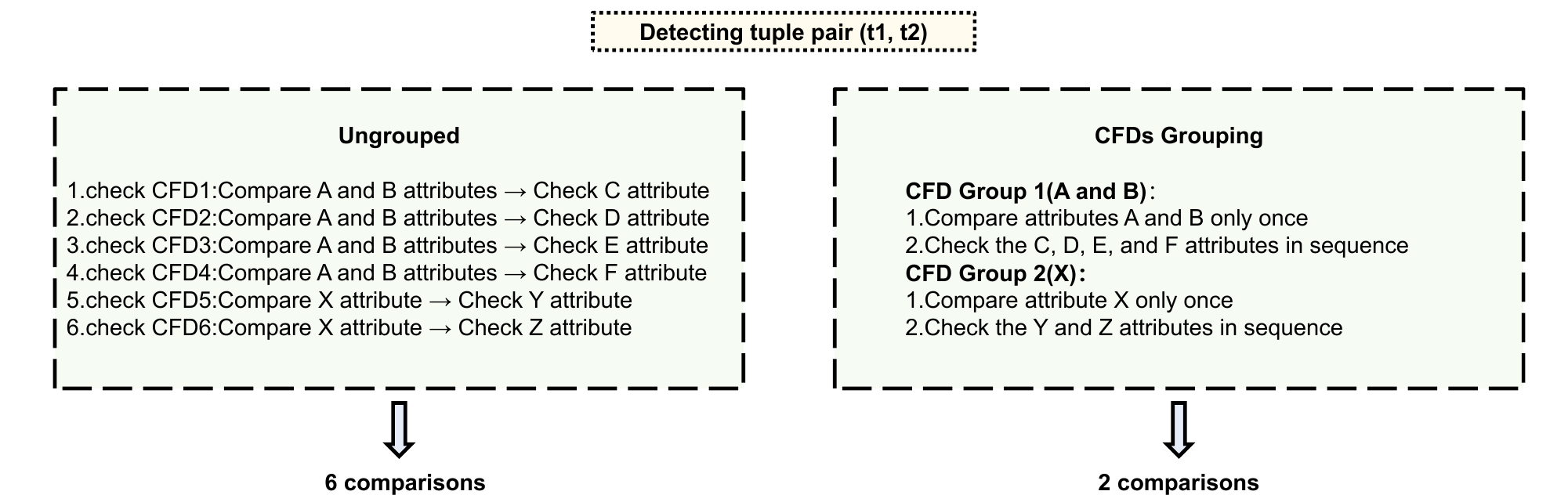}
	\caption{Schematic diagram of CFDs grouping}
	\label{f4-1}
\end{figure}

The CFDs Grouping mechanism effectively eliminates redundant attribute comparison operations, reducing the repeated detection complexity from $O(m \times n)$ to $O(g \times n)$ and further reducing the overall time complexity from $O(m n^2)$ to $O(g n^2)$, where $g$ denotes the number of rule groups ($g \ll m$), thereby significantly improving conflict detection efficiency. Although the CFDs grouping mechanism incurs an additional memory overhead of $O(m)$, its lightweight storage design for group structures only requires maintaining attribute set indices and rule pointers, effectively avoiding data redundancy. Therefore, the rule grouping mechanism essentially serves as an optimization trade-off, achieving significant performance gains in runtime efficiency at the cost of minimal additional space.

\textbf{Conflict graph construction. }The conflict detection results are stored in an undirected graph represented by an adjacency list structure, which serves as the foundational data structure for subsequent connected component analysis. Each conflicting tuple is represented as a vertex in the graph, and each conflicting pair generates a bidirectional edge connecting the corresponding vertices. Adjacency relationships are maintained via a hash map, enabling $O(1)$ complexity for neighbor queries.

\textbf{Algorithm Analysis. }CFDs Grouping and the Attribute Inverted Index form a synergistic optimization framework. The Attribute Inverted Index first enables efficient candidate set filtering, rapidly excluding irrelevant tuple pairs and substantially reducing the detection base. Subsequently, CFDs Grouping is applied to rapidly narrow the detection scope and perform batch rule verification. For each candidate tuple set, conflict detection for all CFDs within a group is completed in a single pass, and the detection of any rule conflict immediately triggers early termination for the current group, significantly reducing redundant attribute matching and rule checking overhead. This hierarchical and progressive optimization effectively and collaboratively addresses the dual bottlenecks of data scale explosion (reduced from $O(n^2)$ to $O(C \cdot n)$) and rule scale explosion (reduced from $O(r)$ to $O(g)$). While preserving detection accuracy, it further optimizes the overall complexity from the traditional $O(r \cdot n^2)$ to $O(g \cdot C \cdot n)$, achieving a leap-level improvement in conflict detection efficiency. The pseudocode for conflict detection is shown in Algorithm 1.

\begin{algorithm}[!t]
\small
\caption{Conflict Detection with Inverted Index and CFDs Grouping}
\label{alg:conflict_detection}
\begin{algorithmic}[1]
\renewcommand{\algorithmicrequire}{\textbf{Input:}}
\renewcommand{\algorithmicensure}{\textbf{Output:}}
\Require Dataset $DB$, CFD list $cfdList$
\Ensure $\mathcal{R}_{conf}$, $\mathcal{P}_{conf}$, $\mathcal{G}_{conf}$, $\mathcal{R}_{Non-conf}$, $\mathcal{M}_{conf}$
  \State $\mathcal{R}_{conf}$: Conflict tuple set $\gets \emptyset$
  \State $\mathcal{P}_{conf}$: Conflict pair list $\gets []$
  \State $\mathcal{G}_{conf}$: Conflict graph $\gets \langle V, E \rangle$ where $V \gets \emptyset, E \gets \emptyset$
  \State $\mathcal{R}_{Non-conf}$: Non-conf tuple set $\gets \emptyset$
  \State $\mathcal{M}_{conf}$: Conflict adjacency map $\gets \{\}$ 

\State $cfdGroups \gets groupCFDs(cfdList)$ 
\State $invIndex \gets buildInvertedIndex(DB)$ 

\For{each tuple $t_1$ at index $i$ in $DB$}
    \State $candidates \gets \emptyset$
    \State $isNon-conf \gets \text{true}$

    \For{each attribute $j$ in $t_1$}
        \State $val \gets value(t_1, j)$
        \State $candidates \gets candidates \cup \{k \mid k \in invIndex[j][val] \land k > i\}$
    \EndFor
    
    \For{each candidate $t_2$ at index $j$ in $candidates$}
        \If{$violatesAnyCFD(t_1, t_2, cfdGroups)$}
            \State $isNon-conf \gets \text{false}$
            
            \If{$i \notin \mathcal{R}_{conf}$}
                \State $\mathcal{R}_{conf} \gets \mathcal{R}_{conf} \cup \{i\}$
            \EndIf
            \If{$j \notin \mathcal{R}_{conf}$}
                \State $\mathcal{R}_{conf} \gets \mathcal{R}_{conf} \cup \{j\}$
            \EndIf
            \State $\mathcal{P}_{conf} \gets \mathcal{P}_{conf} \cup \{(i, j)\}$
            \State $\mathcal{G}_{conf}.V \gets \mathcal{G}_{conf}.V \cup \{i, j\}$
            \State $\mathcal{G}_{conf}.E \gets \mathcal{G}_{conf}.E \cup \{(i, j)\}$
            \State $\mathcal{M}_{conf}[i] \gets \mathcal{M}_{conf}[i] \cup \{j\}$
            \State $\mathcal{M}_{conf}[j] \gets \mathcal{M}_{conf}[j] \cup \{i\}$
        \EndIf
    \EndFor
    
    \If{$isNon-conf \land i \notin \mathcal{R}_{conf}$}
        \State $\mathcal{R}_{Non-conf} \gets \mathcal{R}_{Non-conf} \cup \{i\}$
    \EndIf
\EndFor

\State \textbf{return} $\langle \mathcal{R}_{conf}, \mathcal{P}_{conf}, \mathcal{G}_{conf}, \mathcal{R}_{Non-conf}, \mathcal{M}_{conf} \rangle$
\end{algorithmic}
\end{algorithm}

\begin{example}
\normalfont
Consider the dataset in Table~\ref{employee dataset} and the following three CFDs:
\begin{align*}
\text{CFD}_1: &[\,\text{Work experience}=1, \text{Position}=\text{"manage"}] \rightarrow [\text{Allowance}=1000],\\
\text{CFD}_2: &[\,\text{Work experience}=3, \text{Position}=\text{"technology"}] \rightarrow [\text{Allowance}=800],\\
\text{CFD}_3: &[\,\text{Work experience}=3, \text{Position}=\text{"operate"}, \text{Salary}=8000] \rightarrow [\text{Allowance}=800]
\end{align*}

As shown in Figure~\ref{f4-2}, before the actual detection begins, the CFDs are first grouped according to their left-hand side attributes, and an inverted index is constructed for each attribute value. According to the left-hand side (LHS) attributes, the CFDs in this example are divided into two groups:
$CFD\ group_1$: $CFD_1$ and $CFD_2$ (both share the LHS attributes $Work\ experience$ and $Position$);
$CFD\ group_2$: $CFD_3$ (its LHS attributes are $Work\ experience$, $Position$, and $Salary$).

The conflict detection process begins with an outer loop that iterates over all tuples (the current tuple is denoted as $r_1$). For each $r_1$, a candidate conflict tuple set is constructed, which includes all tuples sharing the same value with $r_1$ on any attribute. Taking $t_0$ as an example:
\begin{align*}
Work\ experience &= 1 \rightarrow [t_1, t_2, t_3] \quad (\text{excluding itself}) \\
Salary &= 6500 \rightarrow [t_1] \\
Position &= manage \rightarrow [t_1, t_2] \\
Allowance &= 1000 \rightarrow [t_1, t_2, t_9]
\end{align*}

After merging and deduplication, the candidate conflict tuple set of $t_0$ is $[t_1, t_2, t_3, t_9]$. In this way, conflict pair detection for $t_0$ only needs to be performed within its candidate conflict tuple set, rather than against all tuples, effectively reducing unnecessary comparisons and improving efficiency.

After constructing the candidate conflicting tuple set, the inner loop iterates over each tuple in this set (denoted as $r_2$). For each tuple pair $(r_1, r_2)$, conflict detection is performed across all CFD groups. Taking $(t_0, t_1)$ as an example:
First, the rules in CFD group 1 are checked. Since $(t_0, t_1)$ matches the left-hand side attribute values of $CFD_1$, a conflict is detected and recorded. Then, for $CFD_2$, as $(t_0, t_1)$ does not match the left-hand side attribute values, $CFD_2$ is skipped. Next, the rules in CFD group 2 are checked, and since $(t_0, t_1)$ does not match the left-hand side attribute values of $CFD_3$, $CFD_3$ is also skipped.
Finally, the results are recorded, and both the conflict list and the conflict mapping table are updated.

\begin{figure}[t]
    \centering
    \includegraphics[width=1\textwidth, height=0.6\textheight, keepaspectratio]{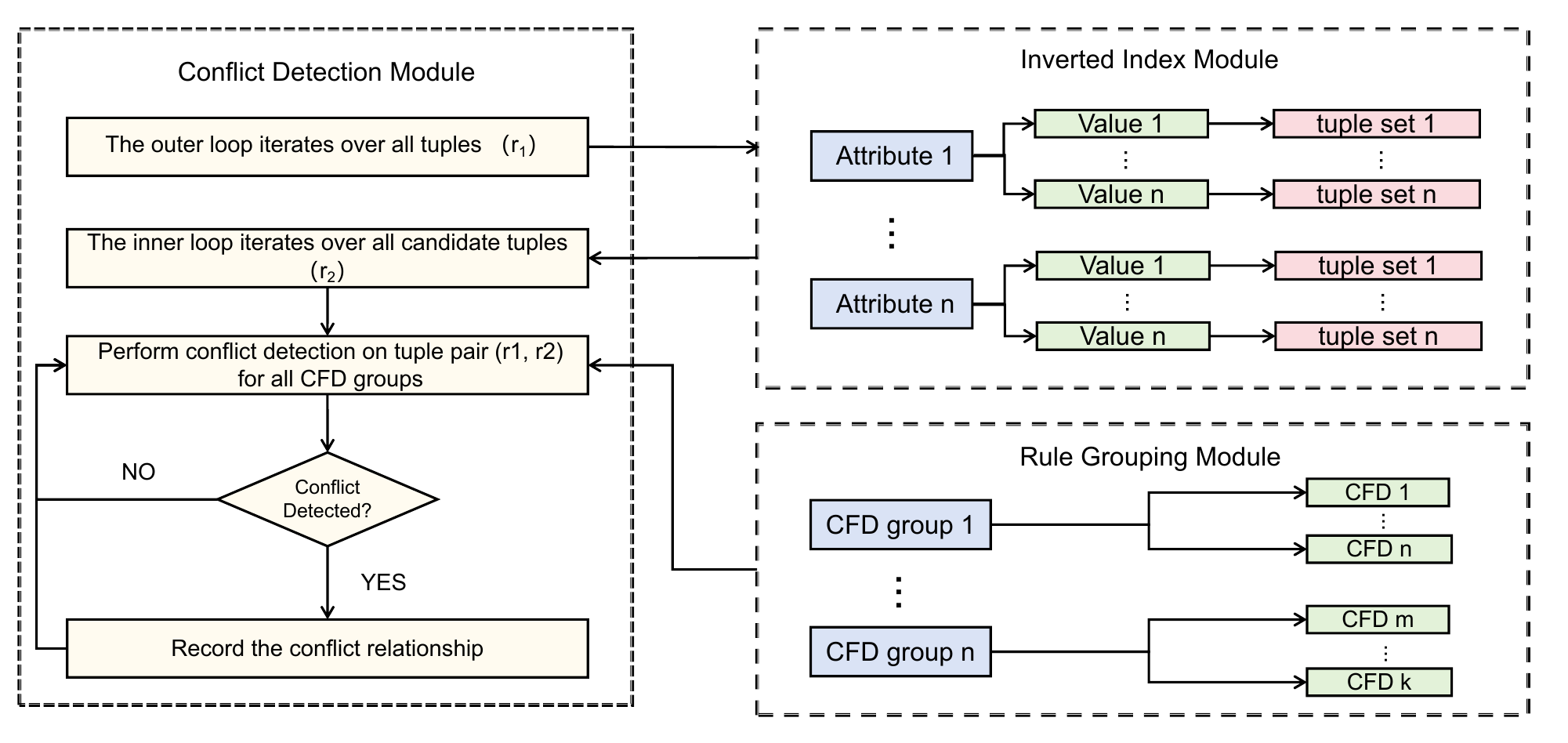}
    \caption{Diagram of Conflict Detection Process}
    \label{f4-2}
\end{figure}
\end{example}

\subsection{Calculation of Density-Conflict Degree Penalty Score}

\textbf{Entropy \& CFD Weighted Density. }EntroCFDensity is a density metric that integrates information entropy and CFDs, used to quantify the quality of tuples in a database. To optimize similarity calculation efficiency in high-dimensional spaces, this work adopts a similarity matrix precomputation strategy, significantly reducing the time complexity of subsequent modules such as density calculation and KNN search. First, the similarity matrix precomputation constructs a two-dimensional array of size $n \times n_c$ ($n_c$ being the number of non-conflicting tuples, i.e., tuples that do not violate any CFD), which contains the similarity between every tuple and each non-conflicting tuple, thereby eliminating redundant similarity calculations in subsequent density computations. Based on this similarity matrix, the $K$ nearest non-conflicting tuple neighbors (kNN) are determined for each tuple, and their approximate density contribution values are calculated accordingly. The process primarily consists of three stages.

The first stage involves weighting each attribute. Specifically, for attribute $i$, its frequency of occurrence in the CFDs is calculated, where $L_r$ denotes the set of attributes on the left-hand side of a CFD, $R_r$ denotes the set of attributes on the right-hand side, $\Sigma$ is the given set of CFDs, and $\mathbb{I}$ is the indicator function:
$$f_i = \sum_{r \in \Sigma} \left[ \mathbb{I}(i \in L_r) + \mathbb{I}(i \in R_r) \right]$$

Next, the frequency weights are normalized, where $\alpha \in (0, 1)$ is an empirical coefficient that can be set according to the specific application:
$$w_i^{\text{CFD}} = \alpha \cdot \frac{f_i}{\max_{j \in U}f_j}$$

Then, the information entropy of attribute $i$ is calculated, followed by the normalization of the entropy weight, where $\beta \in (0, 1)$ is an empirical coefficient and $Values$ represents the set of all possible values for attribute $i$:
$$H_i = -\sum_{v \in \text{Values}(i)} p_v \log p_v, \quad p_v = \frac{\text{count}(v)}{n}$$
$$w_i^{\text{Entropy}} = \beta \cdot \frac{H_i}{\sum_{j=1}^{\lvert U \rvert} H_j}$$

Finally, the comprehensive weight of each attribute is calculated as ($\alpha$ and $\beta$ satisfy $\alpha + \beta = 1$):
$$w_i = w_i^{\text{CFD}} + w_i^{\text{Entropy}} = \alpha \cdot \frac{f_i}{\max_{j \in U}f_j} + \beta \cdot \frac{H_i}{\sum_{j=1}^{\lvert U \rvert} H_j}$$

We set a lower bound of 0.1 for the weights with the following objectives:
(1) Ensure basic attribute participation: prevent any attribute weight from dropping to 0, ensuring that all attributes contribute to some extent.
(2) Avoid extreme bias: in special cases (e.g., when an attribute does not appear in any CFD rule), ensure that the algorithm does not completely ignore the attribute.
(3) Maintain similarity metric validity: ensure that similarity calculations remain multi-dimensional, avoiding reduction to a single-attribute similarity.

The second stage is similarity computation. For each neighbor $t_i \in N_k(t)$, the weighted similarity is calculated as:
$$\text{Sim}(t, t_i) = \sum_{j=1}^{d} w_j \cdot s_j(t, t_i)$$

The similarity $s_j$ is defined as:
$s_j(t, t_i) = 
\begin{cases} 
\frac{1}{1 + \left| v_t^{(j)} - v_{t_i}^{(j)} \right|}, & \text{numerical attribute} \\
\mathbb{I}_{v_t^{(j)} = v_{t_i}^{(j)}}, & \text{categorical attribute} \\
\mathbb{I}\left(v_t^{(j)}\ \text{and}\ v_{t_i}^{(j)}\ \text{are missing}\right), & \text{missing value}
\end{cases}$

The final stage is density aggregation, which aggregates the weighted similarity of $K$ nearest neighbors of each tuple as the density value of that tuple:
$$\rho(t) = \sum_{i=1}^{K} \text{Sim}(t, t_i)$$

\textbf{Conflict degree statistics. } In this study, we introduce a new concept: conflict degree. As defined in Definition 3.8, the conflict degree of a tuple $t_i$ is the number of tuples that conflict with it. The conflict degree captures the global conflict network structure and directly quantifies the extent to which a tuple disrupts overall consistency. Intuitively, tuples with higher conflict degrees are more likely to contain erroneous information, as they are inconsistent with multiple other tuples. Figure~\ref{f4-3} illustrates the relationship between tuples violating constraints and their conflict degrees across four different real-world datasets. It can be observed that dirty data (i.e., data containing errors or inaccurate information) generally exhibits higher conflict degrees than clean data. Based on these observations, our analysis incorporates conflict degree in addition to density. While data density reflects only local neighborhood relationships, conflict degree compensates for this limitation by revealing complex interactions and the global network structure within the dataset.

\begin{figure}[thb]               
	\centering
	\includegraphics[scale=0.33]{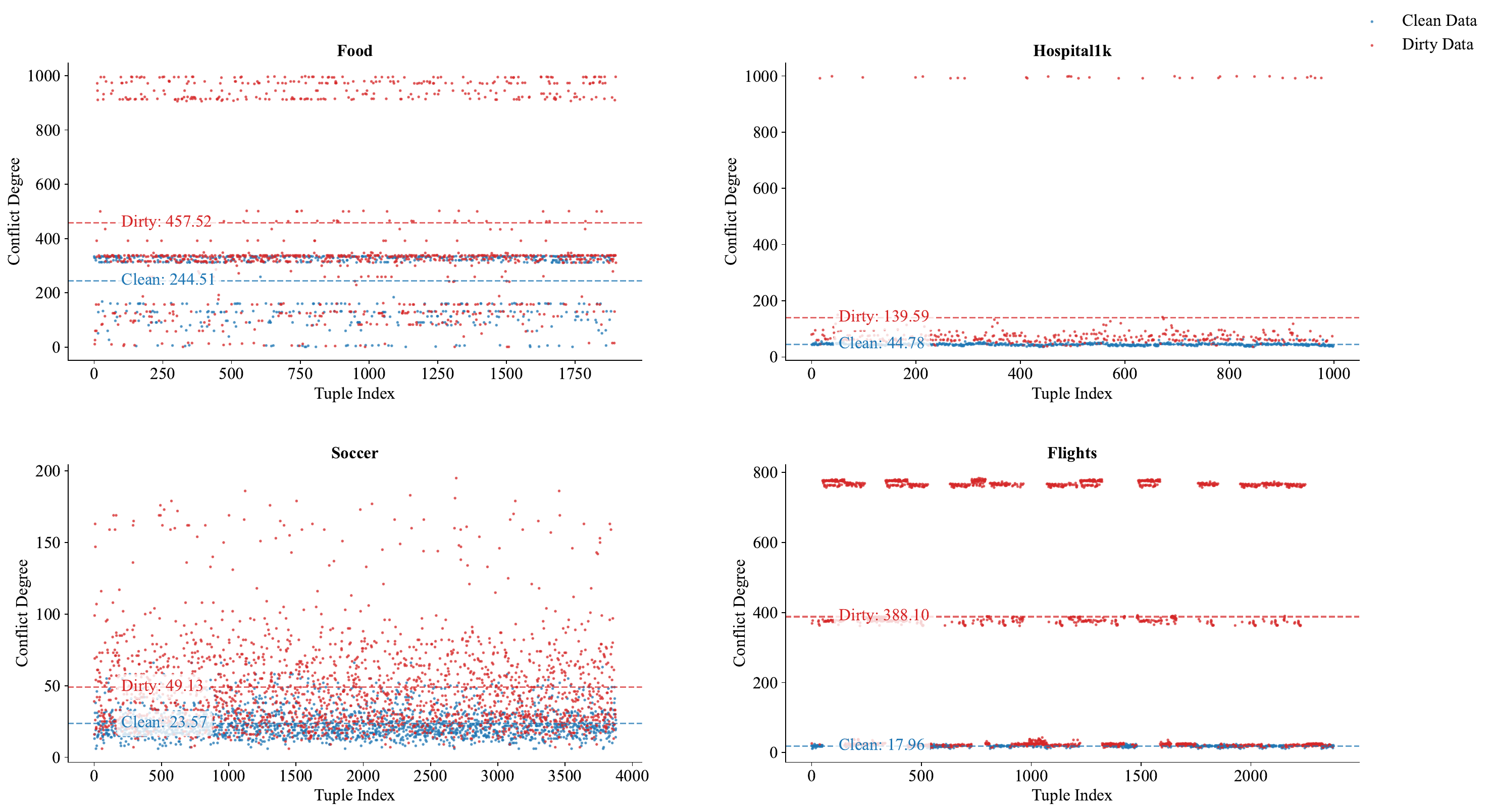}
	\caption{The conflict degree between dirty data and clean data in six datasets}
	\label{f4-3}
\end{figure}

\textbf{Calculation of the density–conflict degree penalty score. } In conflict-graph-based data cleaning algorithms, the penalty score is a key metric that determines the removal priority of tuples. Its design integrates data quality characteristics (density) with conflict topology features (conflict degree). The penalty mechanism guides the generation of the removal set by quantifying the "likelihood of error" for each tuple. While density can capture local similarity, it has inherent limitations in constraint-driven data cleaning, as it cannot identify global constraint violations. In cases where conflicting tuples have high local density, purely density-based methods may mistakenly remove legitimate tuples that have low density but satisfy constraints. By computing a joint score of density and conflict degree (detailed in Definition 3.8), the approach enhances the topological awareness of S-repair. The correction mechanism works as follows: for low-density, non-conflicting tuples, the conflict degree component neutralizes the penalty to avoid false deletions; for high-density, conflicting tuples, the conflict degree component reinforces the penalty to ensure removal. The penalty score formula integrates data quality (density) and conflict topology (conflict degree):

$$\text{penalty}(t) = w_1 \cdot \frac{1}{\rho(t) + \varepsilon} + w_2 \cdot \text{CD}(t)$$

The weights $w_1$ and $w_2$ for weighted density and conflict degree are determined by a coefficient-of-variation-driven adaptive mechanism. These weights are dynamically adjusted based on the statistical characteristics of the connected components. The computation proceeds in three stages.

The first stage is the calculation of the coefficient of variation. $CV_d$ and $CV_c$ represent the coefficients of variation for density and conflict degree, respectively, and are defined as the ratio of the standard deviation ($\sigma$) to the mean ($\mu$). The coefficient of variation standardizes the degree of dispersion, effectively eliminating dimensional differences between variables (density being continuous and conflict degree being discrete), enabling horizontal comparison of data distribution variability. This metric characterizes the heterogeneity of the local data structure and reflects the unevenness and complexity of feature distribution within a region:
$$CV_d = \frac{\sigma_d}{\mu_d}, \quad CV_c = \frac{\sigma_c}{\mu_c}$$

The second stage is the calculation of weights. The base weight is set to 0.5 to represent a neutral decision baseline. The linear amplification of dispersion is achieved using the following formula:
$$\omega'_1 = 0.5 \times (1 + CV_d), \quad \omega'_2 = 0.5 \times (1 + CV_c)$$

The final stage is normalization, whose core objective is to ensure the rationality and validity of the weights. The process consists of two key steps: first, a boundary constraint formula is applied to restrict the weight range, preventing decision bias caused by extreme weights (less than $0.1$ or greater than $0.9$); second, a probabilistic normalization formula is adopted to maintain the interpretability of the weights as probabilities with a total sum of $1$, thereby enhancing the robustness of the algorithm in anomalous components. This design not only avoids decision errors resulting from unreasonable weights but also ensures the scientific soundness and stability of weight allocation. The bounded constraint prevents degeneracy where one feature dominates entirely.
\[
\omega_1 = \max(0.1, \min(0.9, \omega'_1)), \quad
\omega_2 = \max(0.1, \min(0.9, \omega'_2))
\]
\[
\omega_{\text{density}} = \frac{\omega_1}{\omega_1 + \omega_2}, \quad
\omega_{\text{conflict}} = \frac{\omega_2}{\omega_1 + \omega_2}
\]

The dynamic weighting mechanism proposed in this study achieves topology-adaptive decision optimization through a three-stage computational process. First, the coefficient of variation quantifies the dispersion characteristics of data distribution within each connected component. Then, initial weight estimates are generated based on the degree of dispersion. Finally, stable weights are obtained through boundary constraints and probabilistic normalization. This mechanism overcomes the limitations of traditional fixed-weight schemes, enabling a refined "one component, one strategy" decision paradigm—when the density distribution becomes more divergent (high $CV_d$), the algorithm automatically enhances density orientation (favoring the retention of highly consistent tuples); when the conflict distribution becomes more uneven (high $CV_c$), it adaptively strengthens conflict sensitivity (focusing on the removal of frequently violating tuples).

\begin{example}
\normalfont

Consider the example dataset in Table~\ref{employee dataset}, which is subject to the functional dependency set $\Sigma = \{ FD_1: \text{Work experience} \rightarrow \text{Salary},\;
FD_2: \text{Position} \rightarrow \text{Allowance} \}
$. Under these constraints, there exist three candidate minimal removal sets: $S_1=\{t_2,\allowbreak t_3,\allowbreak t_6,\allowbreak t_7,\allowbreak t_8,\allowbreak t_9\}$, $S_2=\{t_2,t_3,t_4,t_5,t_8,t_9\}$, and $S_3=\{t_0,t_1,t_6,t_7,t_8,t_9\}$. In this dataset, the tuples ${t_6,t_7,t_8,t_9}$ form a locally dense dirty cluster because they share the same incorrect values (Work experience $=3$ and Salary $=8000$). Consequently, retaining this cluster leads to a higher aggregated density, even though the cluster itself consists entirely of erroneous tuples. Due to their strong local similarity, these tuples appear highly reliable under density-based evaluation.

In density-based S-repair strategies~\cite{DBLP:journals/tkde/SunSY24}, tuples are evaluated according to their $k$-nearest-neighbor density. In this example, the clean tuples $t_4$ and $t_5$ partially overlap with the dirty cluster on the left-hand side attributes of the dependency and are therefore close to the dirty cluster in the attribute space. As a result, the density-based approach tends to preserve this dense dirty cluster while removing tuples with lower density. This proximity causes the density-based method to mistakenly identify $t_4$ and $t_5$ as erroneous. Consequently, the algorithm selects $S_2$ as the removal set, incorrectly deleting the clean tuples $t_4$ and $t_5$, while retaining erroneous tuples that are mutually reinforced through local density. This phenomenon highlights a fundamental limitation of density-based metrics: when errors appear in clustered forms, locally dense dirty tuples can dominate the aggregation process and lead to incorrect cleaning decisions.

In contrast, the density–conflict degree–based S-repair method evaluates tuples not only according to their local density but also incorporates conflict degree to capture their participation in the global conflict graph. Although the tuples in the dirty cluster exhibit high local density, they are involved in multiple conflicts with respect to $FD_2$ and therefore have high conflict degrees. Conversely, the tuples $t_4$ and $t_5$ participate in relatively few conflicts, and thus are not misclassified as erroneous, even though they are close to dirty data in terms of similarity. As a result, the density–conflict degree–based S-repair selects $S_1$ for removal, correctly identifying the core erroneous cluster for deletion while preserving clean tuples such as $t_4$ and $t_5$.

This example illustrates that density-based metrics are susceptible to locally dense dirty clusters, which may dominate the aggregation process and lead to the misclassification of clean data. By introducing conflict degree, the proposed penalty score shifts the decision criterion from local similarity to global conflict structure, enabling more accurate identification of truly erroneous tuples. Overall, this example demonstrates that incorporating conflict degree effectively mitigates the adverse effects of local density aggregation, resulting in more accurate and robust cleaning decisions.
\end{example}

\textbf{Complexity analysis of density calculation. }
The essence of density computation lies in a multi-stage weighted process that integrates CFD rules and information entropy, maintaining scalability in both time and space dimensions.
In the attribute weighting stage, all rules and tuple attributes are traversed, resulting in a time complexity of $O(|R|d + nd)$, where $|R|$ denotes the number of CFD rules, $d$ is the number of attributes, and $n$ is the total number of tuples. The space consumption of this stage is only $O(d)$.
The similarity matrix is constructed during the initialization phase to establish a complete mapping of pairwise similarities, reducing subsequent similarity query operations in the core algorithm from $O(d)$ to $O(1)$. The time complexity of this stage is $O(n \cdot n_c \cdot d)$, where $n_c$ represents the number of non-conflicting tuples. This design incurs an additional space cost of $O(n \cdot n_c)$, which can be efficiently handled by modern hardware resources in large-scale data cleaning scenarios, providing an order-of-magnitude computational acceleration.
The $K$-nearest neighbor search replaces global sorting with partial sorting, optimizing the time complexity from $O(n \cdot n_c \log n_c)$ to $O(n \cdot n_c \log K)$. The $K$NN storage maintains $K$ nearest neighbor indices and similarities for each tuple, with a space cost of $O(n \cdot K)$.
Finally, the density aggregation stage reuses precomputed weights for similarity accumulation, achieving linear computation in $O(n \cdot K \cdot d)$ time.
In summary, the overall time complexity of the density computation is:
$T(n) = O(|R|d + nd)_{\text{weights}} + O(n \cdot n_c \cdot d)_{\text{similarity}} + O(n \cdot n_c \log K)_{\text{KNN}} + O(n \cdot K \cdot d)_{\text{density}} \approx O(n \cdot n_c)$.

\textbf{Complexity analysis of penalty score calculation for clique components. }
In a clique component, all tuples conflict with each other.
Regarding time complexity, by exploiting the constant property of the conflict degree ($d_c = |C| - 1$), the dynamic statistical stage can be omitted, where $C$ denotes the size of the connected component.
Furthermore, by adopting a fixed weight configuration ($w_1 = 1$, $w_2 = 0$), the computation of the coefficient of variation is skipped, simplifying the penalty score transformation into a pure reciprocal density operation.
As a result, the overall computational complexity is $O(|C|)$.
In terms of space complexity, only the density value array ($O(|C|)$) and constant parameter storage are required.

\textbf{Complexity analysis of penalty score calculation for Non-clique components. }
An efficient processing procedure is achieved through a refined four-stage workflow.
First, the pre-stored density values are retrieved in $O(1)$ time based on the density array. Then, conflict degrees are computed via hash mapping with $O(1)$ complexity.
Subsequently, a single-pass statistical method is employed to calculate the coefficients of variation for both density and conflict degree, enabling dynamic weight assignment within $O(|C|)$ time.
Finally, the penalty aggregation formula is executed, and the initial solution set is constructed in $O(|C|)$ time.
The memory consumption primarily arises from three core data structures: the array storing precomputed density values, the hash map recording conflict relations, and the temporary arrays used to cache statistical quantities of density and conflict degree during coefficient computation.
All three structures strictly maintain a space complexity of $O(|C|)$.

\section{Graph-Decomposition Based Approximate S-repair Algorithms}\label{sec:algorithm1}

Following the construction of the conflict graph, this section further discusses how to utilize the Density-Conflict Degree Penalty Score on the conflict graph to solve the approximate optimal subset repair problem. Unlike directly performing global optimization on the entire conflict graph, the proposed solution framework adheres to a "structure-first, penalty-score-driven" principle, namely: decompose first, score next, and finally solve. Specifically, the conflict graph is first decomposed into several independent connected components, thereby reducing the global repair task into multiple independent subproblems.  Subsequently, the Density–Conflict Degree Penalty Score is computed within each component to guide the repair decision, which effectively avoids the exponential complexity of global optimization.

\begin{figure}[b]
    \centering
    \includegraphics[width=1\textwidth, height=1\textheight, keepaspectratio]{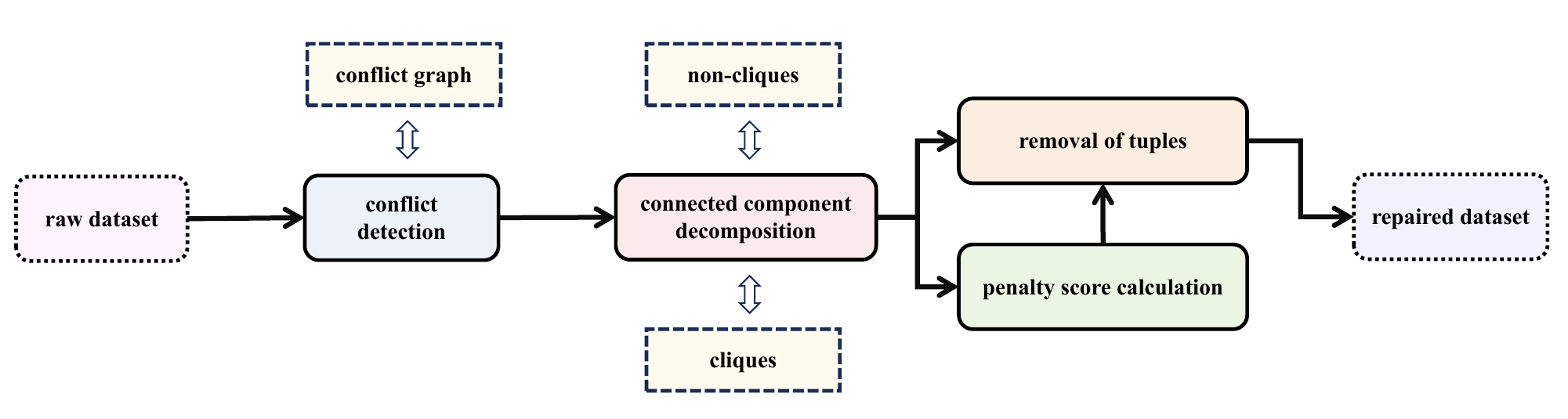}
    \caption{Flowchart of Topology-Aware Subset Repair via Entropy-Guided Density and Graph
Decomposition}
    \label{f5-1}
\end{figure}

Overall, this section proposes an approximate subset repair framework based on graph decomposition and penalty-driven mechanisms, and introduces the high-efficiency PPIS algorithm and the MICO algorithm with theoretical guarantees. 

 As shown in Figure~\ref{f5-1}, the process is primarily divided into four steps: (1) Conflict detection, which efficiently identifies tuple pairs violating CFD rules using inverted indexes and CFDs grouping, followed by the construction of a conflict graph and conflict relationship mapping. (2) Decomposition of the conflict graph into multiple connected components. (3) A penalty score is computed for each tuple based on EntroCFDensity and conflict degree, where the weighting of these two factors is dynamically adjusted via a coefficient of variation-driven adaptive mechanism according to the statistical characteristics of the current connected component. (4) Based on the penalty scores, tuples are removed from each connected component. The union of removal sets from all connected components constitutes the global minimal removal set. Specifically, for clique components, only the tuple with the lowest penalty is retained; for non-clique components, two removal strategies are proposed: Penalty-Prioritized Independent Set (PPIS) and Mixed Integer Covering Optimization (MICO).

\subsection{Connected Component Decomposition of Conflict Graph}

In data cleaning tasks based on CFDs, the topological structure of the conflict graph encapsulates crucial error distribution characteristics. We propose a divide-and-conquer framework for connected component processing: by partitioning the global conflict graph into several topologically independent connected components and independently solving for the minimal removal set within each component,  the framework effectively reduces problem complexity and significantly improves overall cleaning efficiency. Based on the theoretical properties of connected components in graph theory, we further derive the following key conclusions:

\begin{lemma}[Subproblem independence]\label{lem:subproblem-independence}
The global minimal removal set $I_n$ can be decomposed as
$$I_n = \bigcup_{i=1}^{m} I_i,$$
where $I_i$ is the local removal set of connected component $G_i$ and $m$ is the total number of connected components.
\end{lemma}

\begin{proof}
According to Definition~3.5, there are no conflict edges between tuples belonging to different connected components:
$$\forall u \in G_i,\ v \in G_j\ (i \ne j):\ (u,v) \notin E.$$
Therefore, removing a tuple in $G_i$ cannot resolve any conflict in $G_j$, and vice versa.  
Hence the global optimization problem decomposes into $m$ independent subproblems, one for each component.  
Let $I_i$ be the minimal removal set that resolves all conflicts within $G_i$. Any global solution must remove at least $I_i$ from each component, and the union $\bigcup_{i=1}^m I_i$ resolves all conflicts in $G$.  
\end{proof}

To enhance the scalability and computational efficiency of the subsequent repair process, we adopt the iterative depth-first search (IDFS) strategy during the connected component decomposition stage to decompose the conflict graph $G = (V, E)$ obtained from the conflict detection phase into several topologically independent connected subgraphs. Specifically, the algorithm explicitly maintains a processing stack to perform DFS traversal, thereby avoiding the potential stack overflow risks caused by the reliance on the function call stack in traditional recursive implementations. This iterative approach provides higher memory safety and stability when dealing with large-scale datasets or highly connected conflict graphs. Through this decomposition, the original repair problem is reduced to multiple smaller-scale, mutually independent subproblems, laying the foundation for subsequently solving the minimal repair set within each connected component.

\textbf{Complexity Analysis.} In the connected component decomposition process, the time complexity primarily consists of two phases: iterative Depth-First Search (DFS) traversal and topological detection of connected components: (1) Iterative DFS Phase: Requires traversing all nodes and edges in the conflict graph, with a time complexity of $O(|V| + |E|)$, where $|V|$ represents the number of conflicting tuples and $|E|$ represents the number of conflict relationships. (2) Topological Detection Phase: Detects clique components by verifying whether the node degrees satisfy the condition $|C| - 1$ (without needing to traverse specific neighbors), with a time complexity of $O(|V|)$. Overall, the total time complexity is $O(|V| + |E|)$, as both stages exhibit linear complexity characteristics.

The space complexity is determined by the core data structures: (1) Adjacency list storage: Maintaining the topological structure of the conflict graph requires $O(|V| + |E|)$ space. (2) Iterative DFS stack: In the worst case (e.g., a chain-like graph), it requires $O(|V|)$ space. Overall, the total space complexity is $O(|V| + |E|)$, dominated by the adjacency list storage.

\subsection{Removal of tuples}
\textbf{Removal of clique components. }When processing connected components of the conflict graph, it is noted that clique-structured components exhibit complete conflict characteristics, forming a complete subgraph where all tuples are pairwise mutually exclusive. In such structures, at most one tuple can be retained to satisfy constraint consistency. Based on the minimal repair principle, we adopt a penalty-driven strategy: directly retain the tuple with the smallest penalty score within the clique and remove all other tuples. Furthermore, since conflict edges exist between any two vertices within the clique, all tuples in the clique share the same conflict degree. Therefore, the penalty calculation for clique components ignores the conflict degree and uses only density as the metric. This strategy efficiently resolves fully conflicting cases while avoiding unnecessary computation, with an implementation time complexity of $O(n)$, where $n$ denotes the number of nodes within the component, thereby significantly improving processing efficiency.

For connected components with non-clique components, the internal topological relationships are generally more complex and cannot be directly resolved through simple methods. To address this, we propose two solutions: Penalty-Prioritized Independent Set (PPIS) and Mixed Integer Covering Optimization (MICO), which targets different optimization objectives of computational efficiency and repair accuracy, respectively.

\textbf{Penalty-Prioritized Independent Set. } PPIS is designed from the perspective of computational efficiency and constructs a maximal independent set guided by penalty prioritization.  PPIS sorts tuples according to their penalty scores, preferentially selecting nodes with lower penalties to include in the independent set, and iteratively removes nodes that conflict with them. PPIS achieves a nearly linear time complexity, with its time complexity being penalty scores calculation $O(n)$ + sorting $O(n \log n)$ + greedy removal $O(n + m)$ = $O(n \log n + m)$, where $n$ represents the number of nodes within the component and $m$ represents the number of edges in the component. It is suitable for large-scale data cleaning scenarios, capable of providing high-quality approximate repair solutions within an acceptable time frame. The pseudocode is detailed in Algorithm 2.

\begin{algorithm}[!t]
\small
\caption{Penalty-Prioritized Independent Set}
\label{alg:component_removal}
\begin{algorithmic}[1]
\Require 
  $\mathcal{G}_{conf}$: conflict graph, 
  $\mathcal{M}_{conf}$: adjacency map,
  $density[]$: tuple density values
\Ensure $\mathcal{R}_{remove}$: tuples to remove

\State $\mathcal{C} \gets \text{ConnectedComponents}(\mathcal{G}_{conf})$
\State $\mathcal{R}_{remove} \gets \emptyset$

\For{each $comp \in \mathcal{C}$}
    \If{$\text{IsClique}(comp, \mathcal{G}_{conf})$}
        \State $v_{max} \gets \arg\max_{v \in comp} density[v]$
        \State $\mathcal{R}_{remove} \gets \mathcal{R}_{remove} \cup (comp \setminus \{v_{max}\})$
    \Else
        \State Calculate $\mu_d, \sigma_d$ from $\{density[v] : v \in comp\}$
        \State Calculate $\mu_c, \sigma_c$ from $\{|\mathcal{M}_{conf}[v]| : v \in comp\}$
        \State $CV_d \gets \sigma_d / \mu_d$, $CV_c \gets \sigma_c / \mu_c$
        \State $w_1 \gets 0.5 \times (1 + CV_d)$, $w_2 \gets 0.5 \times (1 + CV_c)$
        \State Normalize $w_1, w_2$ to sum 1, clamp to $[0.1, 0.9]$
        \State For each $v \in comp$: $penalty[v] \gets w_1 \cdot \frac{1}{density[v]+\epsilon} + w_2 \cdot |\mathcal{M}_{conf}[v]|$
        \State Sort $comp$ by $penalty$ ascending
        \State $independentSet \gets \emptyset$, $localRemove \gets \emptyset$
        \For{each $v$ in sorted $comp$}
            \If{$\exists u \in \mathcal{M}_{conf}[v] \cap independentSet$}
                \State $localRemove \gets localRemove \cup \{v\}$
            \Else
                \State $independentSet \gets independentSet \cup \{v\}$
            \EndIf
        \EndFor
        \State $\mathcal{R}_{remove} \gets \mathcal{R}_{remove} \cup localRemove$
    \EndIf
\EndFor
\State \textbf{return} $\mathcal{R}_{remove}$
\end{algorithmic}
\end{algorithm}

\begin{example}
\normalfont
As shown in Figure~\ref{f5}, consider a conflict graph $G=(V,E)$ containing 10 tuples. This study adopts a divide-and-conquer strategy to optimize the data cleaning process, and the procedure is as follows. First, the global conflict graph is decomposed into several connected components using Iterative Depth-First Search, and a density value is precomputed for each tuple and stored in a density array. Different processing strategies are applied based on the type of connected component: For connected components that are cliques, a linear scan of the density array is performed to identify the tuple with the highest density within the clique. Only the tuple $t_0$ with the highest density is retained, while the remaining tuples $\{t_3, t_8, t_4\}$ are moved into the $\texttt{removedSet}$, yielding the minimal removal set for this component. For connected components with a non-clique component, the PPIS method is applied. Specifically, a penalty score is first calculated for each tuple $t_i$. The tuples within the component are then sorted in ascending order of their penalty scores, and the $\texttt{independentSet}$ and $\texttt{removedSet}$ are initialized. Finally, a greedy selection is performed via sequential traversal to progressively construct the maximal independent set and determine the component's minimal removal set,the process is as follows:
\begin{align}
& \text{Move } t_1 \text{ into } independentSet. \\
& \text{Move } t_9 \text{ into } independentSet \ (\text{since } t_9 \text{ does not conflict with } t_1). \\
& \text{Move } t_2 \text{ into } removedSet \ (\text{since } t_2 \text{ conflicts with } t_1). \\
& \text{Move } t_5 \text{ into } removedSet \ (\text{since } t_5 \text{ conflicts with } t_9). \\
& \text{Move } t_6 \text{ into } removedSet \ (\text{since } t_6 \text{ conflicts with both } t_1 \text{ and } t_9). \\
& \text{Move } t_7 \text{ into } independentSet \ (\text{since } t_7 \text{ does not conflict with any tuples in the independent set, i.e., } t_1 \text{ and } t_9).
\end{align}

Finally, the minimal removal set generated by this component is $[t_2, t_5, t_6]$.
By merging all local minimal removal sets from each connected component, the global minimal removal set is obtained as $\{t_2, t_3, t_4, t_5, t_6, t_8\}$.

\begin{figure}[thb]
    \centering
    \includegraphics[width=1\textwidth, height=0.6\textheight, keepaspectratio]{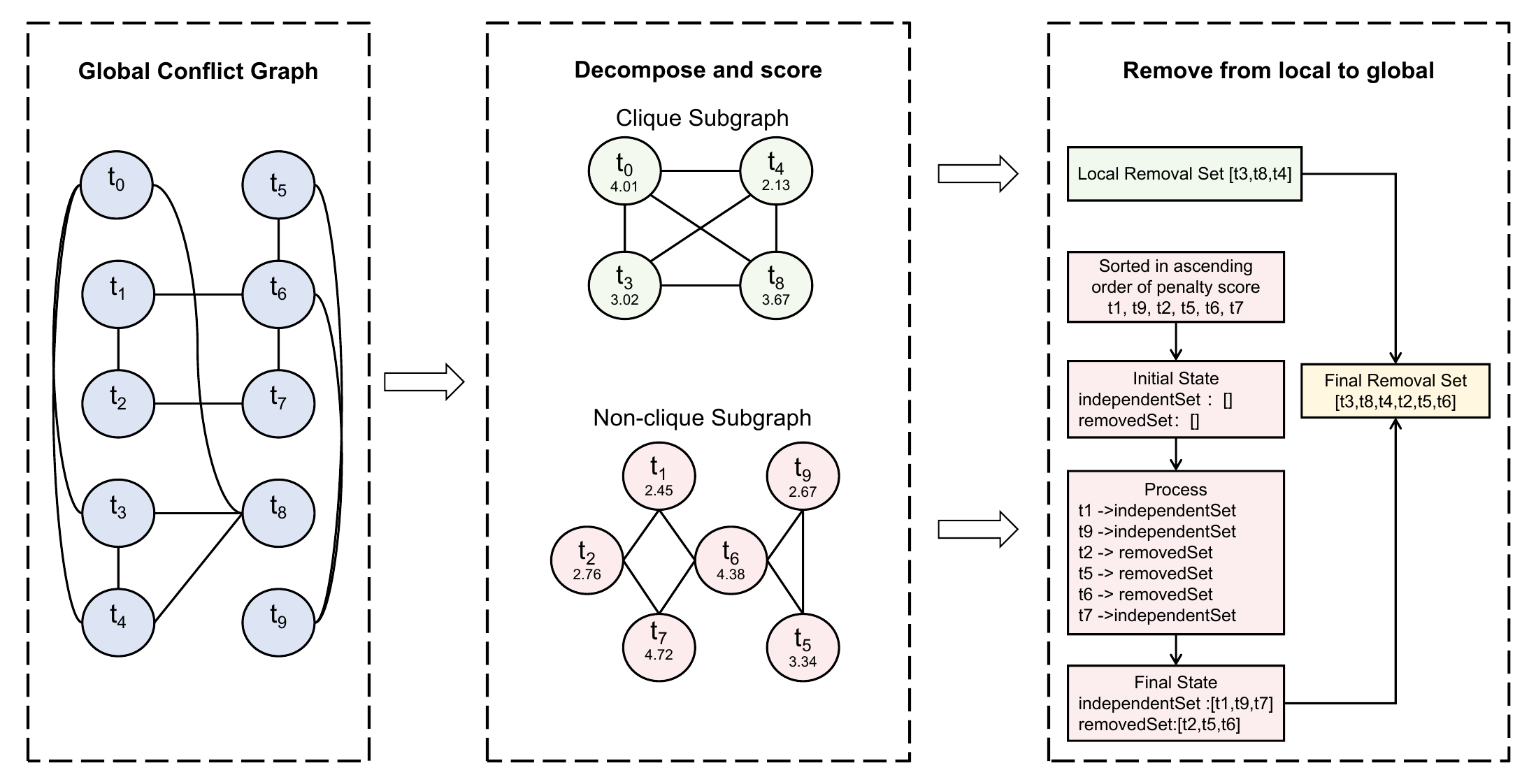}
    \caption{Illustration of Dividing-and-Conquering Removal for clique components and PPIS for non-clique components}
    \label{f5}
\end{figure}

\end{example}

Note that, as an efficient heuristic algorithm, PPIS lacks an approximation ratio guarantee and thus has no strict theoretical assurance. The quality of its solutions depends on the characteristics of the data distribution. Although the dynamic weighting mechanism enhances adaptability, it cannot guarantee convergence to the global optimum.

Although PPIS demonstrates significant advantages in computational efficiency, it suffers from two theoretical deficiencies regarding data repair quality: On one hand, it is prone to local optimum traps, as the greedy strategy makes decisions based solely on the current optimum without backtracking or adjustment. On the other hand, it lacks theoretical guarantees, since existing greedy strategies cannot provide formal assurance on the quality of the obtained solutions.

\textbf{Mixed Integer Covering Optimization. }We therefore propose MICO, which employs a graph decomposition-based mixed integer programming framework. It formalizes the repair problem as a minimum weighted vertex cover problem, independently solving optimization for each connected component. By integrating exact algorithms and heuristic strategies, it ensures solution quality, making it suitable for critical scenarios with stringent accuracy requirements. The core idea is to utilize mathematical programming techniques to solve for the minimum-cost removal set within a limited time frame, while incorporating a robust fallback mechanism to address computational complexity challenges.
 
For each connected component $G_c = (V_c, E_c)$, MICO assigns a removal cost to every tuple $v_i \in V_c$ based on its penalty score. Specifically, let $P_{max}$ denote the maximum penalty value within the current component, and let $\epsilon$ be a small positive constant introduced to ensure numerical stability. The cost function for tuple $v_i$ is defined as:
$$Cost(v_i) = P_{max} - Penalty(v_i) + \epsilon$$

In this formulation, tuples with higher penalty scores receive lower cost values, implying a lower retention value. Consequently, such tuples incur a smaller removal cost and are more likely to be eliminated during repair. This mechanism effectively prioritizes the removal of highly suspicious tuples while preserving tuples with higher data quality.

The decision variables satisfy:

\[
x_i = \begin{cases} 
1 \Rightarrow v_i \text{ will be removed} \\
0 \Rightarrow v_i \text{ will be retained}
\end{cases}
\]

The constraint enforces that at least one variable is set to $1$, ensuring that for every conflict edge, at least one endpoint is removed:

\[
x_i + x_j \geq 1, \quad \forall (v_i, v_j) \in E_c
\]

The optimization objective is to minimize the total cost of removing sets:

\[
\min \sum_{i=1}^{n} \mathrm{Cost}(v_i) \cdot x_i
\]

This model is equivalent to the minimum weighted vertex cover problem, whose physical interpretation is to eliminate all conflicts by minimizing the weighted cost, thereby ensuring data consistency. MICO is implemented using the Gurobi optimization solver, with a time threshold $T\_{max}$ set to control computational overhead. A solution is considered successful if it satisfies the following conditions, where $OPTIMAL$ indicates that the solver finds a global optimal solution within the time limit, and $TIME\_LIMIT$ denotes a timeout case that still returns a feasible (suboptimal) solution:
\[
ValidSolution =
\begin{cases}
OPTIMAL & \Rightarrow \text{Global optimal solution} \\
TIME\_LIMIT \land SolCount > 0 & \Rightarrow \text{Feasible solution}
\end{cases}
\]

To handle solution failures such as infeasibility, numerical instability, or other exceptional cases, MICO triggers a greedy fallback strategy: when the model state is $INFEASIBLE$, $NUMERIC_ERROR$, or $TIME_LIMIT$ with $SolCount=0$, MICO will directly apply PPIS to solve the problem.

\begin{algorithm}[!t]
\small
\caption{Mixed Integer Covering Optimization}
\label{alg:component_mip}
\begin{algorithmic}[1]
\Require 
  $\mathcal{G}_{conf}$: conflict graph, 
  $\mathcal{M}_{conf}$: adjacency map,
  $density[]$: tuple density values
\Ensure $\mathcal{R}_{remove}$: tuples to remove

\State $\mathcal{C} \gets \text{ConnectedComponents}(\mathcal{G}_{conf})$
\State $\mathcal{R}_{remove} \gets \emptyset$

\For{each $comp \in \mathcal{C}$}
    \If{$\text{IsClique}(comp, \mathcal{G}_{conf})$}
        \State $v_{max} \gets \arg\max_{v \in comp} density[v]$
        \State $\mathcal{R}_{remove} \gets \mathcal{R}_{remove} \cup (comp \setminus \{v_{max}\})$
    \Else
        \State Calculate $\mu_d, \sigma_d$ from $\{density[v] : v \in comp\}$
        \State Calculate $\mu_c, \sigma_c$ from $\{|\mathcal{M}_{conf}[v]| : v \in comp\}$
        \State $CV_d \gets \sigma_d / \mu_d$, $CV_c \gets \sigma_c / \mu_c$
        \State $w_1 \gets 0.5 \times (1 + CV_d)$, $w_2 \gets 0.5 \times (1 + CV_c)$
        \State Normalize $w_1, w_2$ to sum 1, clamp to $[0.1, 0.9]$
        \State \textbf{try}
        \State \quad Create MIP model with objective: $\min \sum_{v \in comp} (P_{max} - P(v) + \epsilon) \cdot x_v$
        \State \quad Subject to: $x_u + x_v \geq 1$ for each edge $(u,v) \in \mathcal{G}_{conf}.E$ within $comp$
        \State \quad Solve MIP with time limit $T_{max}$
        \If{solution is optimal or feasible within time limit}
            \State $\mathcal{R}_{remove} \gets \mathcal{R}_{remove} \cup \{v \in comp : x_v = 1\}$
        \Else
            \State \textbf{throw} SolverException
        \EndIf
        \State \textbf{catch} (SolverException or Timeout)
        \State $\mathcal{R}_{greedy} \gets \text{PPIS}(comp, \mathcal{G}_{conf}, w_1, w_2)$
        \State $\mathcal{R}_{remove} \gets \mathcal{R}_{remove} \cup \mathcal{R}_{greedy}$
        \State \textbf{end try}
    \EndIf
\EndFor
\State \textbf{return} $\mathcal{R}_{remove}$
\end{algorithmic}
\end{algorithm}

\textbf{Complexity analysis of MICO. }In the integer programming solving stage, the time complexity for constructing variables and constraints is $O(n + m)$. Solving with Gurobi has a theoretical worst-case exponential time complexity of $O(2^n)$, but under the timeout mechanism, the overall time complexity is effectively $O(n + m + 2^{\min(n, T_{\max})})$, where $n$ is the number of nodes in the connected component, $m$ is the number of edges in the connected component, and $T_{max}$ is the timeout limit (which can be set according to practical requirements). Its space complexity is $O(n + m)$, used for storing decision variables and constraints expressions.  When the timeout rollback is triggered, MICO switches to the PPIS greedy strategy. After fallback, the time complexity directly reduces to that of PPIS, namely $O(n \log n + m)$.

\subsection{Approximate boundary analysis of MICO}This section formally analyzes the approximation performance of the algorithm. By relaxing the problem, only non-conflicting tuples are used as neighbors for density computation, resulting in a relaxed density $\rho'(i)$ that replaces the original density $\rho(i)$. According to Proposition 9 in ~\cite{DBLP:journals/tkde/SunSY24}, the density loss is bounded within $k(r + c)(\rho_{\max} - \rho_{\min})$:
$$|p_{(i)} - p'_{(i)}| \leq \varepsilon = k(r + c)(\rho_{\max} - \rho_{\min})$$
The density loss depends on the problem scale parameters $r, c, m, k$, where $r = |I_R|$ and $c = |I_C|$. Here, $I_C$ denotes the set of conflicting tuple groups, and $I_R$ denotes the set of non-conflicting tuple groups contained in the neighbor set. $k$ represents the number of nearest neighbors. $\rho_{\max}$ and $\rho_{\min}$ are the maximum and minimum density values. Therefore, the penalty score error is $|P(i) - P'(i)| = w_1 \left| \frac{1}{\rho(i)} - \frac{1}{\rho'(i)} \right| \approx w_1 \frac{\varepsilon}{\rho_{small}^2}$, where $\rho_{small}$ is the smaller density value between $\rho(i)$ and $\rho'(i)$. For the remaining tuple set $I_S$, the total penalty score error is proportional to the size of the set $|I_S|$ and is controlled by $\eta$:
$$|penalty(I_S) - penalty'(I_S)| \leq |I_S|\eta, \quad \eta = w_1 \frac{\varepsilon}{\rho_{small}^2}$$

When the MIP is successfully solved, for each connected component, an optimal vertex cover based on the relaxed penalty is obtained, i.e., a minimal removal set $I_n$, such that the remaining tuples $I_S$ satisfy:
$$penalty(I_S) \leq penalty'(S_{\mathrm{opt}})$$
where $S_{\mathrm{opt}}$ denotes the optimal remaining set under the relaxed penalty. Considering the penalty error, we have:
$$penalty(I_S) \leq penalty'(I_S) + |I_S|\eta \leq penalty'(S_{\mathrm{opt}}) + |I_S|\eta \leq penalty(S_{\mathrm{opt}}) + |S_{\mathrm{opt}}|\eta + |I_S|\eta$$
Since $|I_S| \le n$ and $|S_{\mathrm{opt}}| \le n$, where $n$ is the total number of tuples in the dataset, we have:
$$penalty(I_S) \leq penalty(S_{\mathrm{opt}}) + 2n\eta$$

Therefore, the approximate ratio is:
$$\frac{penalty(I_S)}{penalty(S_{\mathrm{opt}})} \leq 1 + \frac{2n\eta}{penalty(S_{\mathrm{opt}})}, \eta = w_1 \frac{k(r + c)(\rho_{\max} - \rho_{\min})}{\rho_{small}^2}$$

\section{Performance evaluation}
In this section, we compare our method with state-of-the-art S-repair approaches, namely HEURISTIC and RELAXATION ~\cite{DBLP:journals/tkde/SunSY24}. All experiments were implemented in Java and conducted on a machine equipped with a 2.1 GHz CPU and 64 GB of memory. Our experiments aim to address the following questions:
(1) How does our algorithm compare to existing S-repair methods?
(2) How does the algorithm performance scale with the number of rows and columns?
(3) How does the choice of the neighborhood parameter $k$ affect algorithm performance?
(4) What is the impact of the number of FDs or CFDs on algorithm performance?
\subsection{Experimental setup.}
\textbf{Datasets. }
This study utilizes nine datasets for experimentation, with detailed information provided in Table~\ref{dataset}. Hospital1k and Flights are two real-world datasets obtained from an earlier research project ~\cite{DBLP:conf/sigmod/MahdaviAFMOS019}, both of which include ground-truth labels. The Movies dataset is taken from the Magellan repository ~\cite{magellandata}, with its true labels coming from existing duplicate tuple annotations in the same prior work ~\cite{DBLP:conf/sigmod/MahdaviAFMOS019}. Hospital10k serves as a common benchmark in data cleaning research, where errors such as spelling mistakes were manually introduced ~\cite{visengeriyeva2018metadata}. The Food dataset contains records of food businesses in Chicago, and approximately 3,000 of its tuples were manually labeled to establish ground truth. For the Soccer dataset, which includes information on football players and teams, the true labels were provided by Rammerlaere and Geerts ~\cite{10.14778/3236187.3236193}. Similarly, the Adult dataset is a well-known census dataset from the UCI repository, and it also comes with ground-truth labels supplied by Rammerlaere and Geerts ~\cite{10.14778/3236187.3236193}. Additionally, the Beers and Restaurant datasets were acquired from a previous project ~\cite{DBLP:conf/cikm/NeutatzMA19}. These contain various types of errors, including missing values, field separation issues, and formatting inconsistencies.

We employ the HYFD algorithm for functional dependency discovery. Table~\ref{dataset} presents the statistical characteristics of the nine datasets, which cover a wide spectrum of data scales (row counts ranging from 1,000 to 32,561) and complexities (number of FDs ranging from 21 to 2,926), establishing a solid foundation for evaluating algorithm generalization capabilities. Among these, the Adult dataset is the largest in scale (32,561 rows) and, additionally, has the most conflicting tuples (32,561). The Movies dataset contains the highest number of FDs (2,926), and Hospital10k has the most erroneous tuples (9,597). This heterogeneous design ensures comprehensive experimentation, effectively testing algorithm robustness under high-dimensional and large-scale data conditions.

Table~\ref{Connectivity Statistics} reports the connectivity statistics of the conflict graph, characterizing its global structural properties in terms of the number of vertices and edges, the number of connected components, and the sizes of the largest and smallest connected components, thereby reflecting how conflicts are distributed and aggregated across the dataset. Table~\ref{Dataset Density and Conflict Statistics} presents the density and conflict statistics of the conflict graph, capturing both conflict intensity and local structural characteristics, including the maximum, minimum, and average conflict degrees, as well as the total, maximum, minimum, and average densities, which together provide a comprehensive view of conflict concentration and structural heterogeneity.

\textbf{Evaluation method. } To demonstrate the effectiveness of S-repair, we evaluate the accuracy of error removal using Precision, Recall, and F1-score. The following metrics are defined: $t_p$ is the number of erroneous tuples correctly removed. $f_p$ is the number of clean tuples incorrectly identified as erroneous. $f_n$ is the number of erroneous tuples not flagged as errors.

\begin{align*}
\text{Recall} &= \frac{t_p}{t_p + f_n} \\
\text{Precision} &= \frac{t_p}{t_p + f_p} \\
\text{F1-score} &= \frac{2 \times \text{Precision} \times \text{Recall}}{\text{Precision} + \text{Recall}}
\end{align*}

\begin{table}[h]
\centering
\small
\caption{Datasets Used in the Experiments}
\label{tab:datasets}
\begin{tabular}{lccccc}
\toprule
\textbf{Dataset} & \textbf{Size} & \textbf{Attributes} & \textbf{FDs} & \textbf{Errors} & \textbf{Conflict Tuples} \\
\midrule
Hospital1k & 1,000 & 19 & 199 & 420 & 1,000 \\
Food & 1,900 & 13 & 184 & 1,167 & 1,836 \\
Flights & 2,376 & 7 & 21 & 1,904 & 2,376 \\
Beers & 2,410 & 10 & 48 & 1,484 & 867 \\
Soccer & 3,880 & 10 & 94 & 2,057 & 3,879 \\
Movies & 7,391 & 29 & 2,926 & 708 & 85 \\
Hospital10k & 10,000 & 17 & 212 & 9,597 & 10,000 \\
Restaurants & 28,787 & 16 & 304 & 4,943 & 28,222 \\
Adult & 32,561 & 16 & 78 & 6,115 & 32,561 \\
\bottomrule
\end{tabular}
\label{dataset}
\end{table}

\begin{table}[h]
\centering
\small
\caption{Conflict Graph Connectivity Statistics}
\begin{tabular}{lccccc}
\toprule
\textbf{Dataset} & \textbf{Nodes} & \textbf{Edges} & \textbf{Connected Components} & \textbf{MaxComponent} & \textbf{MinComponent} \\
\midrule
Hospital1k & 1,000 & 42,300 & 1 & 1,000 & 1,000 \\
Food & 1,836 & 346,197 & 4 & 1,776 & 3 \\
Flights & 2,376 & 373,709 & 1 & 2,376 & 2,376 \\
Beers & 867 & 1,080 & 102 & 62 & 2 \\
Soccer & 3,879 & 72,003 & 2 & 3,871 & 8 \\
Movies & 85 & 441 & 8 & 41 & 2 \\
Hospital10k & 10,000 & 3,123,473 & 1 & 10,000 & 10,000 \\
Restaurants & 28,222 & 503,690 & 1 & 28,222 & 28,222 \\
Adult & 32,561 & 4,588,986 & 1 & 32,561 & 32,561 \\
\bottomrule
\end{tabular}
\label{Connectivity Statistics}
\end{table}

\begin{table}[!h]
\centering
\small
\caption{Statistical Summary of Datasets}
\begin{tabular}{lcccccc}
\hline
\textbf{Dataset} & \textbf{Max Conf-Degree} & \textbf{Min Conf-Degree} & \textbf{Avg Conf-Degree} & \textbf{Max Density} & \textbf{Min Density} & \textbf{Avg Density} \\
\hline
Hospital1k   & 999.000 & 36.000 & 84.600 & 3.000 & 0.3598 & 0.6774 \\
Food         & 998.000 & 1.000  & 377.120 & 0.3949 & 0.0000 & 0.0030 \\
Flights      & 784.000 & 11.000 & 314.570 & 3.000 & 0.0002 & 0.0353 \\
Beers        & 61.000  & 1.000  & 2.490   & 0.8803 & 0.0152 & 0.2381 \\
Soccer       & 195.000 & 6.000  & 37.120  & 3.000 & 0.0776 & 0.3551 \\
Movies       & 27.000  & 1.000  & 10.380  & 1.622 & 0.3198 & 0.6924 \\
Hospital10k  & 2301.000 & 36.000 & 624.690 & 3.000 & 0.0000 & 0.2755 \\
Restaurants  & 2441.000 & 1.000  & 35.690  & 1.049 & 0.0000 & 0.2784 \\
Adult        & 10465.000 & 8.000  & 281.870 & 3.000 & 0.1173 & 1.0305 \\
\hline
\end{tabular}
\label{Dataset Density and Conflict Statistics}
\end{table}

\subsection{Comparison with existing methods. }
Table~\ref{prf} presents a comparison of Precision, Recall, and F1-score across four algorithms on different datasets. The experiments controlled variables with the neighborhood parameter $k$ set to 3 for all algorithms. 

Overall, RELAXATION demonstrates high precision across most datasets, but its recall is generally low. This indicates that RELAXATION is effective at identifying correct positive instances but misses a significant number of true positives. It is more suited to datasets where precision is critical; however, its significant sacrifice in recall limits its overall performance. HEURISTIC shows fluctuating performance across different datasets, with an overall performance that is generally inferior to the other three algorithms. Its stability and generalization ability are relatively limited. The performance of the PPIS algorithm is more stable, maintaining a good balance across different datasets. It performs well in both precision and recall. MICO performs well on multiple datasets and achieves the highest overall F1 score. However, on the Adult dataset, the number of tuples removed by MICO (7,447) significantly exceeds that of the other algorithms, causing its precision and F1 score to be notably lower than PPIS and RELAXATION. This behavior is primarily due to the highly dense conflict graph and weaker constraints in the Adult dataset. Under the hard consistency constraint that all conflict edges must be covered, the MIP solver tends to choose more aggressive covering schemes, deleting a large number of tuples to completely eliminate conflicts, thereby trading higher recall for significant over-cleaning. This suggests that MICO is better suited for scenarios with relatively regular and decomposable conflict graph structures, whereas PPIS, as the default repair strategy, is more robust for highly saturated datasets with complex noise distributions like Adult.

Nevertheless, MICO remains strong across most datasets, achieving a good balance between precision and recall. Therefore, despite its limitations on the Adult dataset, MICO's robustness and effectiveness in other datasets make it a highly stable and efficient algorithm for most application scenarios.

\begin{table}[htbp]
\centering
\caption{Comparison of Precision, Recall, and F1-score for Four Algorithms}
\label{tab:algorithm_comparison}
\scriptsize
\begin{tabular}{lccccccccc}
\toprule
\multirow{2}{*}{Approach} 
& \multicolumn{3}{c}{Food} 
& \multicolumn{3}{c}{Hospital1k} 
& \multicolumn{3}{c}{Hospital10k} \\
\cmidrule(lr){2-4} \cmidrule(lr){5-7} \cmidrule(lr){8-10}
 & Prec. & Rec. & F1 & Prec. & Rec. & F1 & Prec. & Rec. & F1 \\
\midrule
RELAXATION & 0.996 & 0.407 & 0.578 & 0.972 & 0.912 & 0.941 & 0.959 & 0.962 & 0.961 \\
HEURISTIC  & 0.510 & 0.589 & 0.547 & 0.730 & 0.883 & 0.800 & 0.987 & 0.973 & 0.980 \\
PPIS       & 0.979 & 0.405 & 0.573 & 0.969 & 0.886 & 0.925 & 0.964 & 0.922 & 0.943 \\
MICO       & 0.996 & 0.408 & 0.579 & 0.969 & 0.886 & 0.925 & 0.987 & 0.973 & 0.980 \\
\bottomrule
\end{tabular}

\vspace{1em}

\begin{tabular}{lccccccccc}
\toprule
\multirow{2}{*}{Approach} 
& \multicolumn{3}{c}{Movies} 
& \multicolumn{3}{c}{Soccer} 
& \multicolumn{3}{c}{Flights} \\
\cmidrule(lr){2-4} \cmidrule(lr){5-7} \cmidrule(lr){8-10}
 & Prec. & Rec. & F1 & Prec. & Rec. & F1 & Prec. & Rec. & F1 \\
\midrule
RELAXATION & 0.976 & 0.058 & 0.109 & 0.645 & 0.894 & 0.749 & 0.961 & 0.955  & 0.958 \\
HEURISTIC  & 0.909 & 0.057 & 0.106 & 0.627 & 0.887 & 0.735 & 0.961 & 0.953  & 0.957 \\
PPIS       & 0.952 & 0.057 & 0.107 & 0.979 & 0.968 & 0.974 & 0.961 & 0.953  & 0.957 \\
MICO       & 0.976 & 0.058 & 0.109 & 0.973 & 0.961 & 0.967 & 0.955 & 0.910  & 0.932 \\
\bottomrule
\end{tabular}

\vspace{1em}

\begin{tabular}{lccccccccc}
\toprule
\multirow{2}{*}{Approach} 
& \multicolumn{3}{c}{Adult} 
& \multicolumn{3}{c}{Beers} 
& \multicolumn{3}{c}{Restaurants} \\
\cmidrule(lr){2-4} \cmidrule(lr){5-7} \cmidrule(lr){8-10}
 & Prec. & Rec. & F1 & Prec. & Rec. & F1 & Prec. & Rec. & F1 \\
\midrule
RELAXATION & 0.948 & 0.150 & 0.258 &  0.992 & 0.080  & 0.148 & 0.956 & 0.159 & 0.273 \\
HEURISTIC  & 0.253 & 0.330 & 0.287 &  0.976 & 0.083  & 0.153 & 0.863 & 0.157 & 0.266 \\
PPIS       & 0.954 & 0.151 & 0.261 &  0.975 & 0.080  & 0.148 & 0.563 & 0.154  & 0.242 \\
MICO       & 0.253 & 0.330 & 0.287 &  0.959 & 0.079  & 0.146 & 0.934 & 0.155 & 0.265 \\
\bottomrule
\end{tabular}
\label{prf}
\end{table}

Figure~\ref{f6-1} presents the execution times (in milliseconds) of four algorithms. PPIS demonstrates outstanding and consistent efficiency, maintaining the shortest runtime across all datasets, with an average runtime of only 74,236 ms. Its scalability advantages are particularly evident in large-scale scenarios, providing a feasible solution for resource-constrained or real-time demanding applications. In stark contrast, RELAXATION is typically the slowest, with an average runtime as high as 489,567 ms. Especially on datasets with complex constraints, such as Movie, or large-scale datasets like Restaurants, its runtime significantly exceeds that of other algorithms, which limits its practical applicability. HEURISTIC and MICO exhibit moderate efficiency levels, with average runtimes of approximately 289,362 ms and 89,114 ms, respectively. However, the latter often achieves better time efficiency on datasets where it delivers performance comparable to RELAXATION. Comprehensive analysis of both performance and efficiency indicates that MICO effectively controls time costs while maintaining favorable repair performance (e.g., F1 score), achieving a favorable balance between effectiveness and efficiency.

\begin{figure}[h]
    \centering
    \includegraphics[width=1\textwidth, height=0.45\textheight, keepaspectratio]{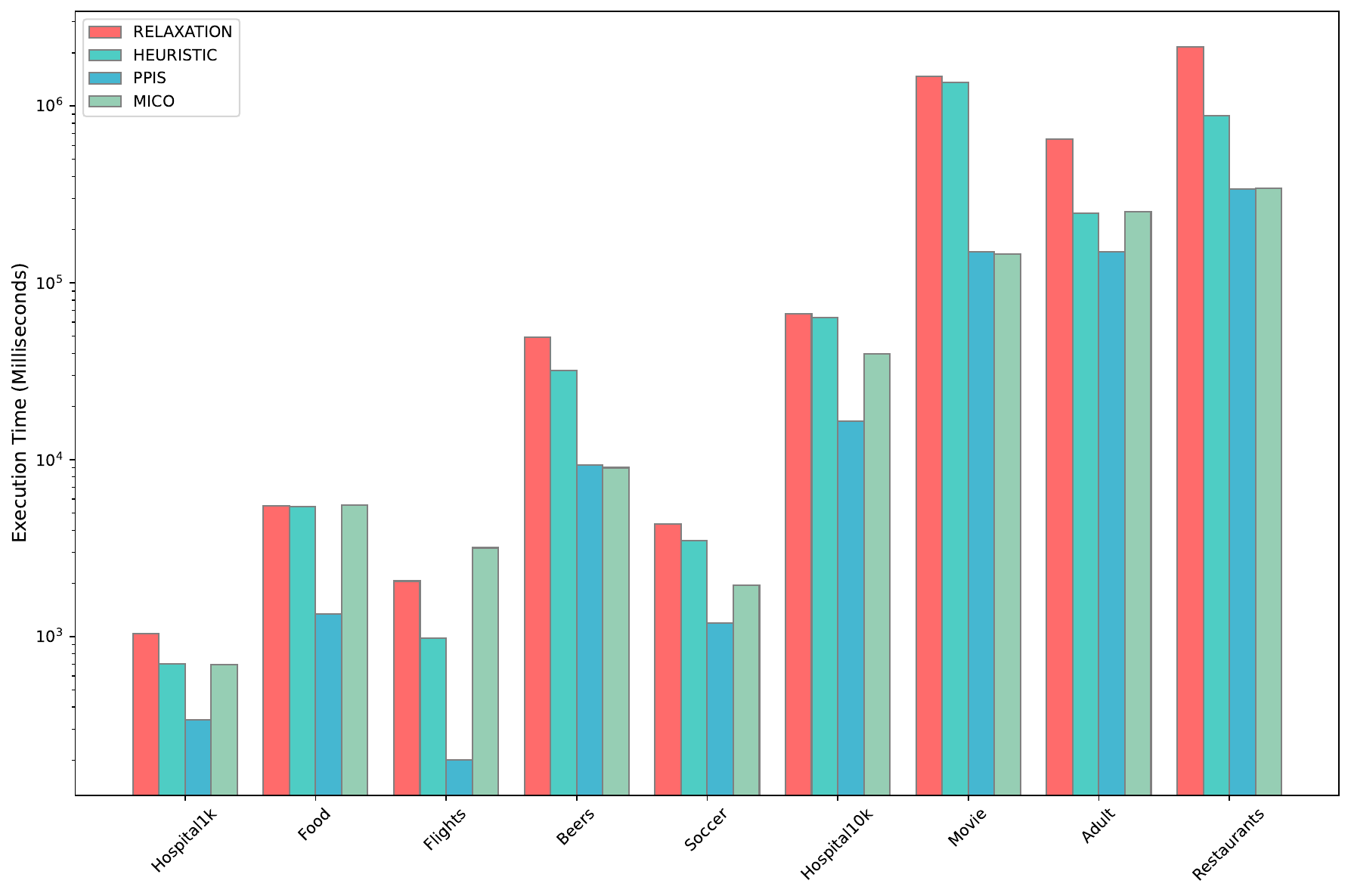}
    \caption{The running time of four algorithms (ms)}
    \label{f6-1}
\end{figure}

Table~\ref{number} compares the number of tuples removed by each algorithm, where the removal count reflects the algorithm's ability to detect and eliminate erroneous tuples. Ideally, this number should approximate the actual number of errors. In most datasets, our proposed PPIS and MICO algorithms tend to reduce unnecessary removals through more precise conflict localization. For example, on the Soccer dataset, both methods remove far fewer tuples (around 2,035) compared with RELAXATION and HEURISTIC (around 2,849), demonstrating higher repair accuracy. On the Adult dataset, however, MICO removes an unusually large number of tuples (7,447), far exceeding the other algorithms. This likely results from a more aggressive covering strategy triggered by the dataset’s complex dependencies and implicit conflicts, which secures high recall but also leads to over-cleaning. Overall, the number of removed tuples is influenced by the error density and conflict structure of each dataset. PPIS and MICO generally show better balancing ability, achieving effective repairs while controlling deletion volume as much as possible, which is valuable in practice for preserving data integrity and usability.
\begin{table}[htbp]
\centering
\caption{Comparison of Removed Tuple Counts}
\label{tab:tuple_removal_comparison}
\scriptsize  
\begin{tabular}{lccccccccc}
\toprule
\textbf{Approach} & \textbf{Food} & \textbf{Hospital1k} & \textbf{Hospital10k} & \textbf{Movies} & \textbf{Soccer} & \textbf{Flights} & \textbf{Adult} & \textbf{Beers} & \textbf{Restaurants} \\
\midrule
RELAXATION & 477 & 394 & 9,629 & 42 & 2,849 & 1891 & 654 & 120 & 824 \\
HEURISTIC  & 1,346 & 508 & 9,460 & 44 & 2,908 & 1888 & 705 & 126 & 900 \\
PPIS       & 483 & 384 & 9,170 & 42 & 2,035 & 1888 & 711 & 122 & 1,350 \\
MICO       & 478 & 384 & 9,460 & 42 & 2,032 & 1814 & 7,447 & 122 & 818 \\
\bottomrule
\end{tabular}
\label{number}
\end{table}

Table~\ref{rate} reports the proportion of clean tuples preserved by each algorithm. The results show that PPIS (average preservation rate: 88.32\%) and MICO (average preservation rate: 91.56\%) generally achieve higher preservation rates than HEURISTIC (average preservation rate: 73.12\%) and RELAXATION (average preservation rate: 80.99\%) across most datasets, especially in those containing clustered errors. The results reveal a fundamental limitation of RELAXATION and HEURISTIC: both rely solely on local density patterns to determine erroneous tuples, meaning that a tuple involved in many local conflicts is automatically treated as suspicious. However, in real-world dirty datasets, errors often appear in clustered forms (dirty clusters), where a group of corrupted tuples densely reference each other. Under such environments, clean tuples located near or partially overlapping these clusters may be incorrectly identified as errors simply due to local overlap. This leads to severe over-cleaning, as seen in Soccer where RELAXATION and HEURISTIC mistakenly remove a large proportion of clean tuples.

In contrast, the Density-Conflict Degree Penalty Score metric employed by PPIS and MICO introduces the concept of conflict degree to capture the global conflict network structure among data, thereby enhancing the perception of data topology and overcoming the limitation of traditional density metrics that only reflect local neighborhood relationships. Specifically, the Density-Conflict Degree Penalty Score exhibits the following corrective characteristics: for tuples with low density but low conflict, the conflict degree term neutralizes the penalty, preventing erroneous deletion. Conversely, tuples with high density but consistently violate multiple constraints or act as central anchors within dirty clusters are assigned higher conflict penalties and are selectively removed.This explains why PPIS and MICO demonstrate better clean data retention capabilities. By quantifying the severity of conflicts, they reduce the risk of large-scale misdeletions, effectively preserving clean tuples that contain correct values but are involved in conflicts due to erroneous tuples.

However, PPIS and MICO do not dominate in all cases. A notable example is the Hospital10k dataset, where HEURISTIC preserves a larger proportion of clean data than PPIS although both are heuristic-driven, and achieves a preservation rate that is comparable to MICO. This behavior is closely tied to the structural properties of the conflict graph. HEURISTIC, relying primarily on local density measurement, tends to retain tuples that participate in only a small number of local conflicts, which in Hospital10k are often clean tuples positioned at the periphery of dense error clusters. In contrast, while PPIS introduces the concept of global conflict degree, its heavy reliance on global conflict information in datasets with complex conflict structures, such as Hospital10k, may lead to misdeletions. Specifically, PPIS tends to delete tuples that are associated with multiple conflicts, including clean data that only partially overlap with error clusters. Due to the overweighting of global conflict degree in such high‑density conflict environments, this can result in the unintended removal of clean data that should have been retained. 

Therefore, the reason HEURISTIC achieves better preservation on Hospital10k likely lies in its focus on local conflict detection rather than comprehensive calculation of global conflict structures. Although PPIS  can identify more complex conflict patterns, it fails to effectively distinguish truly clean data from data that are merely locally affected by errors when the relationship between errors and conflict degree in the conflict graph is weak. This over‑reliance on the global structure of the conflict graph leads to the loss of clean data. MICO, however, behaves differently. As a global optimization framework, its mixed-integer formulation searches the conflict graph more holistically, enabling it to identify a solution structure that, in this particular dataset, more closely resembles the behavior of HEURISTIC. This allows MICO to avoid some of the cascading misdeletions observed in PPIS and to preserve a substantially larger fraction of clean data.

In summary, the effectiveness of PPIS and MICO lies not in conservative deletion, but in their global conflict-aware selectivity achieved through the introduction of conflict degree. By distinguishing core violators from clean tuples merely influenced by local dirty clusters, PPIS and MICO substantially reduce false deletions and provide a robust solution to the over-cleaning problem inherent in traditional S-repair techniques.

\begin{table}[htbp]
\centering
\caption{Clean Data Retention Rate (\%)}
\label{tab:clean_data_retention}
\scriptsize
\begin{tabular}{lccccccccc}
\toprule
\textbf{Approach} & \textbf{Food} & \textbf{Hospital1k} & \textbf{Hospital10k} & \textbf{Movies} & \textbf{Soccer} & \textbf{Flights} & \textbf{Adult} & \textbf{Beers} & \textbf{Restaurants} \\
\midrule
RELAXATION & 99.73 & 98.10 & 2.48 & 99.99 & 44.51 & 84.53 & 99.85 & 99.89 & 99.85 \\
HEURISTIC & 10.10 & 76.38 & 69.98 & 99.94 & 40.50 & 84.53 & 77.51 &  99.68 & 99.48 \\
PPIS & 98.64 & 97.93 & 19.11 & 99.97 & 97.64 & 84.53 & 99.86 & 99.68 & 99.53 \\
MICO & 99.73 & 97.93 & 69.98 & 99.99 & 96.98 & 82.84 & 77.51 & 99.46 & 99.77 \\
\bottomrule
\end{tabular}
\label{rate}
\end{table}

\subsection{The performance of our methods. }
As shown in Figure~\ref{f6-2}, the experimental results indicate that both methods maintain strong stability in F1-score across different values of $k$, with overall performance remaining high and minimally affected by changes in $k$. This demonstrates that both methods exhibit good robustness in terms of accuracy.
Regarding computational efficiency, the execution time of both methods generally increases as $k$ grows. This trend is primarily due to the larger number of candidate elements that the algorithms must process with increasing $k$, which leads to higher time costs. Notably, MICO consistently requires more execution time than PPIS.
In terms of the penalty associated with remaining tuples, the total penalty decreases continuously as $k$ increases for both methods, indicating that larger $k$ values help cover more potentially important tuples and thus reduce the penalty from omissions. This effect is particularly pronounced when $k$ is small, while the rate of penalty reduction gradually slows once $k$ exceeds a certain threshold, reflecting a diminishing marginal benefit.
Overall, the choice of $k$ involves a trade-off between model performance, computational efficiency, and coverage completeness, and should be configured appropriately based on the specific application scenario.
\begin{figure}[h]
    \centering
    \includegraphics[width=1\textwidth, height=1\textheight, keepaspectratio]{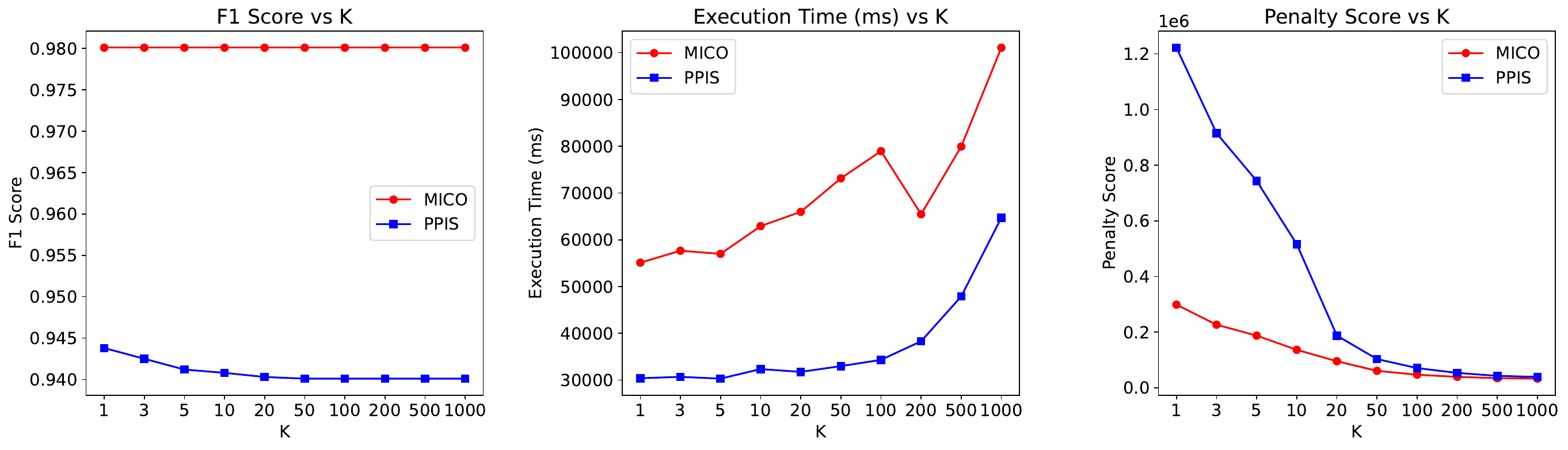}
    \caption{F1, running time, and penalty score at different k values}
    \label{f6-2}
\end{figure}

Figure~\ref{f6-3} illustrates the changes in execution time for the MICO and PPIS algorithms under different parameter settings. The left subplot shows that as the number of columns increases, the execution time of MICO exhibits a fluctuating trend, whereas PPIS demonstrates a relatively stable upward trend, with only minor fluctuations at certain points. The right subplot shows that as the number of rows increases, the execution time of both algorithms continues to rise. However, MICO's execution time grows significantly faster than PPIS, particularly in regions with a high number of rows, where MICO's runtime increases sharply while PPIS maintains a relatively gradual growth. This indicates that PPIS may offer better performance stability when handling large-scale data.

\begin{figure}[h]
    \centering
    \includegraphics[width=1\textwidth, height=0.6\textheight, keepaspectratio]{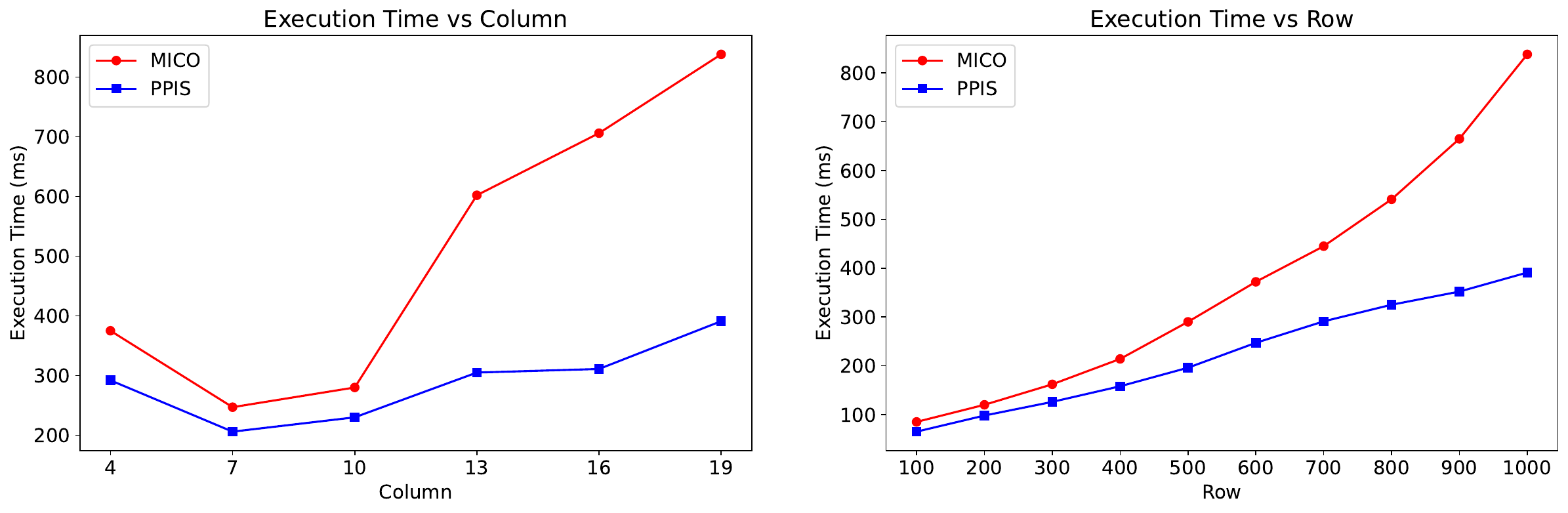}
    \caption{Running time under different rows and columns}
    \label{f6-3}
\end{figure}

\section{Conclusion}
This study addresses the limitations of traditional density-based subset repair methods in scenarios with clustered dirty data, including biased density estimation, computational inefficiency, and homogeneity of attribute weights. We propose an approximately optimal subset repair framework based on a density–conflict degree penalty. By integrating information entropy with conditional functional dependency (CFD) rules, we construct a dynamically weighted density model (EntroCFDensity), and, combined with a conflict-degree metric to quantify each tuple's impact on global consistency, the framework effectively balances sensitivity to local similarity and global conflict topology.

A conflict-graph divide-and-conquer strategy is innovatively designed: during conflict detection, an attribute inverted index and CFD rule grouping mechanism are introduced, significantly reducing detection complexity; during repair, the NP-hard global optimization problem is decomposed into independently processed connected components. For clique components, a penalty-driven optimal retention strategy is applied, while for non-clique components, an efficiency-prioritized PPIS greedy algorithm and a precision-oriented MICO mixed-integer programming method are employed, achieving a coordinated optimization of computational efficiency and repair quality.

Experimental results show that the proposed methods outperform the latest density-based S-repair approaches across multiple datasets. The PPIS algorithm maintains high F1 performance while achieving near-linear time complexity, making it particularly suitable for large-scale data cleaning. The MICO algorithm provides strict theoretical guarantees, enabling reliable decision making in critical application scenarios. Both methods achieve a significantly improved balance between tuple retention and error removal, thereby demonstrating the effectiveness of the conflict-degree metric in correcting bias caused by clustered dirty data. Overall, within the same topology-aware approximate subset-repair framework based on the joint density–conflict penalty score, PPIS offers a high-efficiency greedy solution that scales to large datasets, while MICO provides a mixed-integer programming formulation with stronger theoretical assurances. 

This study provides a new approach for constraint-driven data cleaning that balances efficiency and accuracy. Future work will explore several avenues to further enhance the flexibility, robustness, and applicability of the proposed framework.

\section*{Acknowledgment}
This work was supported in part by National Natural Science Foundation of China grant no. 62402135, U21A20513, Taishan Scholars Program of Shandong Province grant no. tsqn202211091, Shandong Provincial Natural Science Foundation grant no. ZR2023QF059.

\bibliographystyle{plainnat}
\bibliography{reference}

\begin{thebibliography}{37}
\providecommand{\natexlab}[1]{#1}
\providecommand{\url}[1]{\texttt{#1}}
\expandafter\ifx\csname urlstyle\endcsname\relax
  \providecommand{\doi}[1]{doi: #1}\else
  \providecommand{\doi}{doi: \begingroup \urlstyle{rm}\Url}\fi

\bibitem[Abdelaal et~al.(2024)Abdelaal, Ktitarev, St{\"{a}}dtler, and
  Sch{\"{o}}ning]{DBLP:conf/edbt/0001KSS24}
Mohamed Abdelaal, Tim Ktitarev, Daniel St{\"{a}}dtler, and Harald
  Sch{\"{o}}ning.
\newblock {SAGED:} few-shot meta learning for tabular data error detection.
\newblock In Letizia Tanca, Qiong Luo, Giuseppe Polese, Loredana Caruccio,
  Xavier Oriol, and Donatella Firmani, editors, \emph{Proceedings 27th
  International Conference on Extending Database Technology, {EDBT} 2024,
  Paestum, Italy, March 25 - March 28}, pages 386--398. OpenProceedings.org,
  2024.
\newblock \doi{10.48786/EDBT.2024.34}.
\newblock URL \url{https://doi.org/10.48786/edbt.2024.34}.

\bibitem[Armstrong(1974)]{DBLP:conf/ifip/Armstrong74}
William~Ward Armstrong.
\newblock Dependency structures of data base relationships.
\newblock In Jack~L. Rosenfeld, editor, \emph{Information Processing,
  Proceedings of the 6th {IFIP} Congress 1974, Stockholm, Sweden, August 5-10,
  1974}, pages 580--583. North-Holland, 1974.

\bibitem[Bao et~al.(2024)Bao, Binbin, Fan, Li, Li, Lin, Lin, Liu, Liu, Lv,
  Ouyang, Sun, Tang, Wang, Wei, Wu, Xie, Zhang, Zhao, Zhu, and
  Zhu]{DBLP:journals/pvldb/BaoBFLLLLLLLOSTWWWXZZZZ24}
Zian Bao, Bie Binbin, Wenfei Fan, Daji Li, Mengyun Li, Kaiwen Lin, Wei Lin,
  Peijie Liu, Peng Liu, Zhicong Lv, Mingliang Ouyang, Chenyang Sun, Shuai Tang,
  Yaoshu Wang, Qiyuan Wei, Xiangqian Wu, Min Xie, Jing Zhang, Runxiao Zhao, Jie
  Zhu, and Yilin Zhu.
\newblock Rock: Cleaning data with both {ML} and logic rules.
\newblock \emph{Proc. {VLDB} Endow.}, 17\penalty0 (12):\penalty0 4373--4376,
  2024.
\newblock \doi{10.14778/3685800.3685878}.
\newblock URL \url{https://www.vldb.org/pvldb/vol17/p4373-wang.pdf}.

\bibitem[Bohannon et~al.(2007)Bohannon, Fan, Geerts, Jia, and
  Kementsietsidis]{DBLP:conf/icde/BohannonFGJK07}
Philip Bohannon, Wenfei Fan, Floris Geerts, Xibei Jia, and Anastasios
  Kementsietsidis.
\newblock Conditional functional dependencies for data cleaning.
\newblock In Rada Chirkova, Asuman Dogac, M.~Tamer {\"{O}}zsu, and Timos~K.
  Sellis, editors, \emph{Proceedings of the 23rd International Conference on
  Data Engineering, {ICDE} 2007, The Marmara Hotel, Istanbul, Turkey, April
  15-20, 2007}, pages 746--755. {IEEE} Computer Society, 2007.
\newblock \doi{10.1109/ICDE.2007.367920}.
\newblock URL \url{https://doi.org/10.1109/ICDE.2007.367920}.

\bibitem[Breve et~al.(2022)Breve, Caruccio, Deufemia, and
  Polese]{DBLP:conf/edbt/BreveCDP22}
Bernardo Breve, Loredana Caruccio, Vincenzo Deufemia, and Giuseppe Polese.
\newblock {RENUVER:} {A} missing value imputation algorithm based on relaxed
  functional dependencies.
\newblock In Julia Stoyanovich, Jens Teubner, Paolo Guagliardo, Milos Nikolic,
  Andreas Pieris, Jan M{\"{u}}hlig, Fatma {\"{O}}zcan, Sebastian Schelter,
  H.~V. Jagadish, and Meihui Zhang, editors, \emph{Proceedings of the 25th
  International Conference on Extending Database Technology, {EDBT} 2022,
  Edinburgh, UK, March 29 - April 1, 2022}, pages 1:52--1:64.
  OpenProceedings.org, 2022.
\newblock \doi{10.5441/002/EDBT.2022.05}.
\newblock URL \url{https://doi.org/10.5441/002/edbt.2022.05}.

\bibitem[Chen et~al.(2011)Chen, Chen, Conway, Hellerstein, and
  Parikh]{DBLP:journals/tkde/ChenCCHP11}
Kuang Chen, Harr Chen, Neil Conway, Joseph~M. Hellerstein, and Tapan~S. Parikh.
\newblock Usher: Improving data quality with dynamic forms.
\newblock \emph{{IEEE} Trans. Knowl. Data Eng.}, 23\penalty0 (8):\penalty0
  1138--1153, 2011.
\newblock \doi{10.1109/TKDE.2011.31}.
\newblock URL \url{https://doi.org/10.1109/TKDE.2011.31}.

\bibitem[Chomicki and Marcinkowski(2005)]{DBLP:journals/iandc/ChomickiM05}
Jan Chomicki and Jerzy Marcinkowski.
\newblock Minimal-change integrity maintenance using tuple deletions.
\newblock \emph{Inf. Comput.}, 197\penalty0 (1-2):\penalty0 90--121, 2005.
\newblock \doi{10.1016/J.IC.2004.04.007}.
\newblock URL \url{https://doi.org/10.1016/j.ic.2004.04.007}.

\bibitem[Chu et~al.(2013)Chu, Ilyas, and Papotti]{DBLP:conf/icde/ChuIP13}
Xu~Chu, Ihab~F. Ilyas, and Paolo Papotti.
\newblock Holistic data cleaning: Putting violations into context.
\newblock In Christian~S. Jensen, Christopher~M. Jermaine, and Xiaofang Zhou,
  editors, \emph{29th {IEEE} International Conference on Data Engineering,
  {ICDE} 2013, Brisbane, Australia, April 8-12, 2013}, pages 458--469. {IEEE}
  Computer Society, 2013.
\newblock \doi{10.1109/ICDE.2013.6544847}.
\newblock URL \url{https://doi.org/10.1109/ICDE.2013.6544847}.

\bibitem[Dallachiesa et~al.(2013)Dallachiesa, Ebaid, Eldawy, Elmagarmid, Ilyas,
  Ouzzani, and Tang]{DBLP:conf/sigmod/DallachiesaEEEIOT13}
Michele Dallachiesa, Amr Ebaid, Ahmed Eldawy, Ahmed~K. Elmagarmid, Ihab~F.
  Ilyas, Mourad Ouzzani, and Nan Tang.
\newblock {NADEEF:} a commodity data cleaning system.
\newblock In Kenneth~A. Ross, Divesh Srivastava, and Dimitris Papadias,
  editors, \emph{Proceedings of the {ACM} {SIGMOD} International Conference on
  Management of Data, {SIGMOD} 2013, New York, NY, USA, June 22-27, 2013},
  pages 541--552. {ACM}, 2013.
\newblock \doi{10.1145/2463676.2465327}.
\newblock URL \url{https://doi.org/10.1145/2463676.2465327}.

\bibitem[Das et~al.()Das, Doan, Paul~Suganthan, Gokhale, Konda, Govind, and
  Paulsen]{magellandata}
Sanjib Das, AnHai Doan, G.~C. Paul~Suganthan, Chaitanya Gokhale, Pradap Konda,
  Yash Govind, and Derek Paulsen.
\newblock The magellan data repository.
\newblock \url{https://sites.google.com/site/anhaidgroup/projects/data}.
\newblock Accessed: 2025-01-01.

\bibitem[Ding et~al.(2022)Ding, Wang, Su, Wang, Li, and
  Gao]{DBLP:journals/tkde/DingWSWLG22}
Xiaoou Ding, Hongzhi Wang, Jiaxuan Su, Muxian Wang, Jianzhong Li, and Hong Gao.
\newblock Leveraging currency for repairing inconsistent and incomplete data.
\newblock \emph{{IEEE} Trans. Knowl. Data Eng.}, 34\penalty0 (3):\penalty0
  1288--1302, 2022.
\newblock \doi{10.1109/TKDE.2020.2992456}.
\newblock URL \url{https://doi.org/10.1109/TKDE.2020.2992456}.

\bibitem[Geerts et~al.(2013)Geerts, Mecca, Papotti, and
  Santoro]{DBLP:journals/pvldb/GeertsMPS13}
Floris Geerts, Giansalvatore Mecca, Paolo Papotti, and Donatello Santoro.
\newblock The {LLUNATIC} data-cleaning framework.
\newblock \emph{Proc. {VLDB} Endow.}, 6\penalty0 (9):\penalty0 625--636, 2013.
\newblock \doi{10.14778/2536360.2536363}.
\newblock URL \url{http://www.vldb.org/pvldb/vol6/p625-mecca.pdf}.

\bibitem[Hao et~al.(2017)Hao, Tang, Li, He, Ta, and
  Feng]{DBLP:journals/tkde/Hao0LHTF17}
Shuang Hao, Nan Tang, Guoliang Li, Jian He, Na~Ta, and Jianhua Feng.
\newblock A novel cost-based model for data repairing.
\newblock \emph{{IEEE} Trans. Knowl. Data Eng.}, 29\penalty0 (4):\penalty0
  727--742, 2017.
\newblock \doi{10.1109/TKDE.2016.2637928}.
\newblock URL \url{https://doi.org/10.1109/TKDE.2016.2637928}.

\bibitem[Heidari et~al.(2019)Heidari, McGrath, Ilyas, and
  Rekatsinas]{DBLP:conf/sigmod/HeidariMIR19}
Alireza Heidari, Joshua McGrath, Ihab~F. Ilyas, and Theodoros Rekatsinas.
\newblock Holodetect: Few-shot learning for error detection.
\newblock In Peter~A. Boncz, Stefan Manegold, Anastasia Ailamaki, Amol
  Deshpande, and Tim Kraska, editors, \emph{Proceedings of the 2019
  International Conference on Management of Data, {SIGMOD} Conference 2019,
  Amsterdam, The Netherlands, June 30 - July 5, 2019}, pages 829--846. {ACM},
  2019.
\newblock \doi{10.1145/3299869.3319888}.
\newblock URL \url{https://doi.org/10.1145/3299869.3319888}.

\bibitem[Ilyas and Chu(2015)]{DBLP:journals/ftdb/IlyasC15}
Ihab~F. Ilyas and Xu~Chu.
\newblock Trends in cleaning relational data: Consistency and deduplication.
\newblock \emph{Found. Trends Databases}, 5\penalty0 (4):\penalty0 281--393,
  2015.
\newblock \doi{10.1561/1900000045}.
\newblock URL \url{https://doi.org/10.1561/1900000045}.

\bibitem[Koehler and Link(2022)]{DBLP:journals/tkde/KoehlerL22}
Henning Koehler and Sebastian Link.
\newblock Possibilistic data cleaning.
\newblock \emph{{IEEE} Trans. Knowl. Data Eng.}, 34\penalty0 (12):\penalty0
  5939--5950, 2022.
\newblock \doi{10.1109/TKDE.2021.3062318}.
\newblock URL \url{https://doi.org/10.1109/TKDE.2021.3062318}.

\bibitem[Lew et~al.(2021)Lew, Agrawal, Sontag, and
  Mansinghka]{DBLP:conf/aistats/LewASM21}
Alexander~K. Lew, Monica Agrawal, David~A. Sontag, and Vikash Mansinghka.
\newblock Pclean: Bayesian data cleaning at scale with domain-specific
  probabilistic programming.
\newblock In Arindam Banerjee and Kenji Fukumizu, editors, \emph{The 24th
  International Conference on Artificial Intelligence and Statistics, {AISTATS}
  2021, April 13-15, 2021, Virtual Event}, volume 130 of \emph{Proceedings of
  Machine Learning Research}, pages 1927--1935. {PMLR}, 2021.
\newblock URL \url{http://proceedings.mlr.press/v130/lew21a.html}.

\bibitem[Li et~al.(2021)Li, Rao, Blase, Zhang, Chu, and
  Zhang]{DBLP:conf/icde/LiRBZCZ21}
Peng Li, Xi~Rao, Jennifer Blase, Yue Zhang, Xu~Chu, and Ce~Zhang.
\newblock Cleanml: {A} study for evaluating the impact of data cleaning on {ML}
  classification tasks.
\newblock In \emph{37th {IEEE} International Conference on Data Engineering,
  {ICDE} 2021, Chania, Greece, April 19-22, 2021}, pages 13--24. {IEEE}, 2021.
\newblock \doi{10.1109/ICDE51399.2021.00009}.
\newblock URL \url{https://doi.org/10.1109/ICDE51399.2021.00009}.

\bibitem[Liang et~al.(2025)Liang, Su, Song, and
  Li]{DBLP:journals/tkde/LiangSSL25}
Kenny~Ye Liang, Yunxiang Su, Shaoxu Song, and Chunping Li.
\newblock Turn waste into wealth: On efficient clustering and cleaning over
  dirty data.
\newblock \emph{{IEEE} Trans. Knowl. Data Eng.}, 37\penalty0 (7):\penalty0
  4361--4372, 2025.
\newblock \doi{10.1109/TKDE.2025.3564313}.
\newblock URL \url{https://doi.org/10.1109/TKDE.2025.3564313}.

\bibitem[Liu et~al.(2024)Liu, Shen, Ghosh, Gilad, Kimelfeld, and
  Roy]{DBLP:journals/pvldb/LiuSGGKR24}
Yuxi Liu, Fangzhu Shen, Kushagra Ghosh, Amir Gilad, Benny Kimelfeld, and
  Sudeepa Roy.
\newblock The cost of representation by subset repairs.
\newblock \emph{Proc. {VLDB} Endow.}, 18\penalty0 (2):\penalty0 475--487, 2024.
\newblock URL \url{https://www.vldb.org/pvldb/vol18/p475-liu.pdf}.

\bibitem[Liu et~al.(2022)Liu, Zhou, and Rekatsinas]{DBLP:journals/vldb/LiuZR22}
Zifan Liu, Zhechun Zhou, and Theodoros Rekatsinas.
\newblock Picket: guarding against corrupted data in tabular data during
  learning and inference.
\newblock \emph{{VLDB} J.}, 31\penalty0 (5):\penalty0 927--955, 2022.
\newblock \doi{10.1007/S00778-021-00699-W}.
\newblock URL \url{https://doi.org/10.1007/s00778-021-00699-w}.

\bibitem[Livshits et~al.(2020)Livshits, Kimelfeld, and
  Roy]{DBLP:journals/tods/LivshitsKR20}
Ester Livshits, Benny Kimelfeld, and Sudeepa Roy.
\newblock Computing optimal repairs for functional dependencies.
\newblock \emph{{ACM} Trans. Database Syst.}, 45\penalty0 (1):\penalty0
  4:1--4:46, 2020.
\newblock \doi{10.1145/3360904}.
\newblock URL \url{https://doi.org/10.1145/3360904}.

\bibitem[Mahdavi et~al.(2019)Mahdavi, Abedjan, Fernandez, Madden, Ouzzani,
  Stonebraker, and Tang]{DBLP:conf/sigmod/MahdaviAFMOS019}
Mohammad Mahdavi, Ziawasch Abedjan, Raul~Castro Fernandez, Samuel Madden,
  Mourad Ouzzani, Michael Stonebraker, and Nan Tang.
\newblock Raha: {A} configuration-free error detection system.
\newblock In Peter~A. Boncz, Stefan Manegold, Anastasia Ailamaki, Amol
  Deshpande, and Tim Kraska, editors, \emph{Proceedings of the 2019
  International Conference on Management of Data, {SIGMOD} Conference 2019,
  Amsterdam, The Netherlands, June 30 - July 5, 2019}, pages 865--882. {ACM},
  2019.
\newblock \doi{10.1145/3299869.3324956}.
\newblock URL \url{https://doi.org/10.1145/3299869.3324956}.

\bibitem[Miao et~al.(2020)Miao, Cai, Li, Gao, and
  Liu]{DBLP:journals/pvldb/MiaoCLGL20}
Dongjing Miao, Zhipeng Cai, Jianzhong Li, Xiangyu Gao, and Xianmin Liu.
\newblock The computation of optimal subset repairs.
\newblock \emph{Proc. {VLDB} Endow.}, 13\penalty0 (11):\penalty0 2061--2074,
  2020.
\newblock URL \url{http://www.vldb.org/pvldb/vol13/p2061-miao.pdf}.

\bibitem[Miao et~al.(2023)Miao, Zhang, Li, Wang, and
  Cai]{DBLP:journals/vldb/MiaoZLWC23}
Dongjing Miao, Pengfei Zhang, Jianzhong Li, Ye~Wang, and Zhipeng Cai.
\newblock Approximation and inapproximability results on computing optimal
  repairs.
\newblock \emph{{VLDB} J.}, 32\penalty0 (1):\penalty0 173--197, 2023.
\newblock \doi{10.1007/S00778-022-00738-0}.
\newblock URL \url{https://doi.org/10.1007/s00778-022-00738-0}.

\bibitem[Neutatz et~al.(2019)Neutatz, Mahdavi, and
  Abedjan]{DBLP:conf/cikm/NeutatzMA19}
Felix Neutatz, Mohammad Mahdavi, and Ziawasch Abedjan.
\newblock {ED2:} {A} case for active learning in error detection.
\newblock In Wenwu Zhu, Dacheng Tao, Xueqi Cheng, Peng Cui, Elke~A.
  Rundensteiner, David Carmel, Qi~He, and Jeffrey~Xu Yu, editors,
  \emph{Proceedings of the 28th {ACM} International Conference on Information
  and Knowledge Management, {CIKM} 2019, Beijing, China, November 3-7, 2019},
  pages 2249--2252. {ACM}, 2019.
\newblock \doi{10.1145/3357384.3358129}.
\newblock URL \url{https://doi.org/10.1145/3357384.3358129}.

\bibitem[Pena et~al.(2021)Pena, de~Almeida, and
  Naumann]{DBLP:journals/pvldb/PenaAN21}
Eduardo H.~M. Pena, Eduardo~Cunha de~Almeida, and Felix Naumann.
\newblock Fast detection of denial constraint violations.
\newblock \emph{Proc. {VLDB} Endow.}, 15\penalty0 (4):\penalty0 859--871, 2021.
\newblock \doi{10.14778/3503585.3503595}.
\newblock URL \url{https://www.vldb.org/pvldb/vol15/p859-pena.pdf}.

\bibitem[Qin et~al.(2024)Qin, Huang, Wang, Zhu, Zhang, Miao, Mao, Onizuka, and
  Xiao]{DBLP:conf/icde/QinHWZZM0O024}
Jianbin Qin, Sifan Huang, Yaoshu Wang, Jing Zhu, Yifan Zhang, Yukai Miao, Rui
  Mao, Makoto Onizuka, and Chuan Xiao.
\newblock Bclean: {A} bayesian data cleaning system.
\newblock In \emph{40th {IEEE} International Conference on Data Engineering,
  {ICDE} 2024, Utrecht, The Netherlands, May 13-16, 2024}, pages 3407--3420.
  {IEEE}, 2024.
\newblock \doi{10.1109/ICDE60146.2024.00263}.
\newblock URL \url{https://doi.org/10.1109/ICDE60146.2024.00263}.

\bibitem[Rammelaere and Geerts(2018)]{10.14778/3236187.3236193}
Joeri Rammelaere and Floris Geerts.
\newblock Explaining repaired data with cfds.
\newblock \emph{Proc. VLDB Endow.}, 11\penalty0 (11):\penalty0 1387–1399,
  July 2018.
\newblock ISSN 2150-8097.
\newblock \doi{10.14778/3236187.3236193}.
\newblock URL \url{https://doi.org/10.14778/3236187.3236193}.

\bibitem[Rezig et~al.(2021)Rezig, Ouzzani, Aref, Elmagarmid, Mahmood, and
  Stonebraker]{DBLP:journals/pvldb/RezigOAEMS21}
El~Kindi Rezig, Mourad Ouzzani, Walid~G. Aref, Ahmed~K. Elmagarmid, Ahmed~R.
  Mahmood, and Michael Stonebraker.
\newblock Horizon: Scalable dependency-driven data cleaning.
\newblock \emph{Proc. {VLDB} Endow.}, 14\penalty0 (11):\penalty0 2546--2554,
  2021.
\newblock \doi{10.14778/3476249.3476301}.
\newblock URL \url{http://www.vldb.org/pvldb/vol14/p2546-rezig.pdf}.

\bibitem[Sun et~al.(2020)Sun, Song, Wang, and Wang]{DBLP:conf/icde/SunSW020}
Yu~Sun, Shaoxu Song, Chen Wang, and Jianmin Wang.
\newblock Swapping repair for misplaced attribute values.
\newblock In \emph{36th {IEEE} International Conference on Data Engineering,
  {ICDE} 2020, Dallas, TX, USA, April 20-24, 2020}, pages 721--732. {IEEE},
  2020.
\newblock \doi{10.1109/ICDE48307.2020.00068}.
\newblock URL \url{https://doi.org/10.1109/ICDE48307.2020.00068}.

\bibitem[Sun et~al.(2024)Sun, Song, and Yuan]{DBLP:journals/tkde/SunSY24}
Yu~Sun, Shaoxu Song, and Xiaojie Yuan.
\newblock From minimum change to maximum density: On determining near-optimal
  s-repair.
\newblock \emph{{IEEE} Trans. Knowl. Data Eng.}, 36\penalty0 (2):\penalty0
  627--639, 2024.
\newblock \doi{10.1109/TKDE.2023.3294401}.
\newblock URL \url{https://doi.org/10.1109/TKDE.2023.3294401}.

\bibitem[Tsamardinos et~al.(2006)Tsamardinos, Brown, and
  Aliferis]{DBLP:journals/ml/TsamardinosBA06}
Ioannis Tsamardinos, Laura~E. Brown, and Constantin~F. Aliferis.
\newblock The max-min hill-climbing bayesian network structure learning
  algorithm.
\newblock \emph{Mach. Learn.}, 65\penalty0 (1):\penalty0 31--78, 2006.
\newblock \doi{10.1007/S10994-006-6889-7}.
\newblock URL \url{https://doi.org/10.1007/s10994-006-6889-7}.

\bibitem[Tzoumas et~al.(2011)Tzoumas, Deshpande, and
  Jensen]{DBLP:journals/pvldb/TzoumasDJ11}
Kostas Tzoumas, Amol Deshpande, and Christian~S. Jensen.
\newblock Lightweight graphical models for selectivity estimation without
  independence assumptions.
\newblock \emph{Proc. {VLDB} Endow.}, 4\penalty0 (11):\penalty0 852--863, 2011.
\newblock URL \url{http://www.vldb.org/pvldb/vol4/p852-tzoumas.pdf}.

\bibitem[Visengeriyeva and Abedjan(2018)]{visengeriyeva2018metadata}
Larysa Visengeriyeva and Ziawasch Abedjan.
\newblock Metadata-driven error detection.
\newblock In \emph{Proc. Int. Conf. Scientific, Stat. Database Manage. SSDBM},
  pages 1--12, 2018.

\bibitem[Wang and Tang(2014)]{DBLP:conf/sigmod/WangT14}
Jiannan Wang and Nan Tang.
\newblock Towards dependable data repairing with fixing rules.
\newblock In Curtis~E. Dyreson, Feifei Li, and M.~Tamer {\"{O}}zsu, editors,
  \emph{International Conference on Management of Data, {SIGMOD} 2014,
  Snowbird, UT, USA, June 22-27, 2014}, pages 457--468. {ACM}, 2014.
\newblock \doi{10.1145/2588555.2610494}.
\newblock URL \url{https://doi.org/10.1145/2588555.2610494}.

\bibitem[Zhang et~al.(2023)Zhang, Hu, Zong, Li, and
  Xia]{DBLP:conf/apweb/ZhangHZLX23}
Anzhen Zhang, Shengji Hu, Chuanyu Zong, Jiajia Li, and Xiufeng Xia.
\newblock Computing maximal likelihood subset repair for inconsistent data.
\newblock In Xiangyu Song, Ruyi Feng, Yunliang Chen, Jianxin Li, and Geyong
  Min, editors, \emph{Web and Big Data - 7th International Joint Conference,
  APWeb-WAIM 2023, Wuhan, China, October 6-8, 2023, Proceedings, Part {II}},
  volume 14332 of \emph{Lecture Notes in Computer Science}, pages 1--15.
  Springer, 2023.
\newblock \doi{10.1007/978-981-97-2390-4\_1}.
\newblock URL \url{https://doi.org/10.1007/978-981-97-2390-4\_1}.

\end{thebibliography}







        
	


\end{document}